\documentclass[12pt]{article}

\usepackage{amssymb}
\usepackage{amsmath}
\usepackage{amsthm}
\usepackage{upref}
\usepackage{enumerate}

\usepackage[utf8]{inputenc}
\usepackage[english]{babel}
\usepackage{graphicx}
\usepackage{color}

\usepackage[a4paper,margin=2.5cm]{geometry}
\theoremstyle{plain}
\newtheorem{theorem}{Theorem}[section]
\newtheorem{proposition}[theorem]{Proposition}
\newtheorem{lemma}[theorem]{Lemma}

\theoremstyle{definition}

\newtheorem{assumption}[theorem]{Assumption}
\newtheorem{remark}[theorem]{Remark}

\numberwithin{equation}{section}

\usepackage{enumerate}

\newcommand{\abs}[1]{\lvert{#1}\rvert}

\newcommand{\norm}[1]{\lVert{#1}\rVert}

\newcommand{\ip}[2]{\langle{#1},{#2}\rangle}

\DeclareMathOperator*{\slim}{s-lim}

\DeclareMathOperator{\supp}{supp}

\newcommand{\vp}{\varphi}
\newcommand{\bR}{\mathbf{R}}
\newcommand{\bC}{\mathbf{C}}

\newcommand{\bone}{\mathbf{1}}
\newcommand{\bzero}{\mathbf{0}}

\newcommand{\sfone}{\mathsf{1}}
\newcommand{\sfzero}{\mathsf{0}}

\newcommand{\cF}{\mathcal{F}}

\newcommand{\cH}{\mathcal{H}}

\newcommand{\cB}{\mathcal{B}}
\newcommand{\cK}{\mathcal{K}}
\newcommand{\cS}{\mathcal{S}}

\newcommand{\bZ}{\mathbf{Z}}
\newcommand{\bT}{\mathbf{T}}

\newcommand{\sfD}{\mathsf{D}}
\newcommand{\sfS}{\mathsf{S}}

\newcommand{\sfI}{\mathsf{I}}

\newcommand{\Tdh}{\bT^d_h}

\newcommand{\sff}{\mathsf{f}}

\newcommand{\sfF}{\mathsf{F}}

\newcommand{\e}{{\rm e}}
\newcommand{\I}{{\rm i}}
\newcommand{\di}{\,{\rm d}}
\newcommand{\ident}{\ensuremath{I}}
\newcommand{\identh}{\ensuremath{I_h}}
\newcommand{\identhsf}{\mathsf{I}_h}

\newcommand{\mH}{\widetilde{H}_{0,h}}
\newcommand{\nbc}[1]{\lVert{#1}\rVert_{\cB(\bC^2)}}
\newcommand{\nbcc}[1]{\lVert{#1}\rVert_{\cB(\bC^4)}}
\newcommand{\Gfb}{G_{0,h}^{\rm fb}}
\newcommand{\Hfb}{H_{0,h}^{\rm fb}}
\newcommand{\tGfb}{\widetilde{G}_{0,h}^{\rm fb}}
\newcommand{\tHfb}{\widetilde{H}_{0,h}^{\rm fb}}
\newcommand{\gfb}{g_{0,h}^{\rm fb}}
\newcommand{\tgfb}{\tilde{g}_{0,h}^{\rm fb}}
\newcommand{\Hhs}{H^{\rm s}_{0,h}}
\newcommand{\Hhfb}{H^{\rm fb}_{0,h}}
\newcommand{\Ghs}{G^{\rm s}_{0,h}}
\newcommand{\Ghfb}{G^{\rm fb}_{0,h}}
\newcommand{\ghs}{g^{\rm s}_{0,h}}

\newcommand{\tHhs}{\widetilde{H}^{\rm s}_{0,h}}
\newcommand{\tGhs}{\widetilde{G}^{\rm s}_{0,h}}
\newcommand{\tghs}{\tilde{g}^{\rm s}_{0,h}}
\newcommand{\Sobone}{H^1(\bR)\otimes\bC^2}
\newcommand{\Sobtwo}{H^1(\bR^2)\otimes\bC^2}
\newcommand{\Sobthree}{H^1(\bR^3)\otimes\bC^4}

\begin{document}
\title{Discrete approximations to Dirac operators and norm resolvent convergence}
\author{
Horia Cornean\footnote{Department of Mathematical Sciences, Aalborg University, Skjernvej 4A, 
DK-9220 Aalborg \O{}, Denmark. Email: cornean@math.aau.dk and matarne@math.aau.dk},\;
Henrik Garde\footnote{Department of Mathematics, Aarhus University, Ny Munkegade 118, DK-8000 Aarhus C, Denmark. Email: garde@math.au.dk},\; and
Arne Jensen\footnotemark[1]
}
\date{}
\maketitle

\begin{abstract}
	\noindent We consider continuous Dirac operators defined on $\bR^d$, $d\in\{1,2,3\}$, together with various discrete versions of them. Both forward-backward and symmetric finite differences are used as approximations to partial derivatives. We also allow a bounded, H\"older continuous, and self-adjoint matrix-valued potential, which in the discrete setting is evaluated on the mesh.
	Our main goal is to investigate whether the proposed discrete models converge in norm resolvent sense  to their continuous counterparts, as the mesh size tends to zero and up to a natural embedding of the discrete space into the continuous one. In dimension one we show that forward-backward differences lead to norm resolvent convergence, while in dimension two and three they do {\it not}. The same negative result holds in all dimensions when symmetric differences are used. On the other hand, strong resolvent convergence holds in all these cases. Nevertheless, and quite remarkably, a rather simple but non-standard modification to the discrete models, involving the mass term, ensures norm resolvent convergence in general.
\end{abstract}

\section{Introduction}

We study in detail in what sense continuous Dirac operators can be approximated by a family of discrete operators indexed by the mesh size. To investigate spectral properties based on the discrete models, it is essential to know whether we can obtain norm resolvent convergence or only strong resolvent convergence of the discrete models (suitably embedded into the continuum) to the continuous Dirac operators. 

In this paper we present a remarkable new phenomenon. In dimensions two and three we cannot obtain norm resolvent convergence of the discrete operators (embedded into the continuum) as the mesh size tends to zero, if we use the natural discretizations based on either symmetric first order differences or a pair of forward-backward first order differences.
The models require a simple modification to obtain norm resolvent convergence. In dimension one the discretization using a pair of forward-backward first order differences does lead to norm resolvent convergence, whereas the model based on symmetric first order differences does not.

These results are now described in some detail.
To unify the notation we define $\nu(1)=\nu(2)=2$ and $\nu(3)=4$.
The Hilbert spaces used for the continuous Dirac operators are
\begin{equation*}
\cH^d=L^2(\bR^d)\otimes \bC^{\nu(d)},
\quad d=1,2,3.
\end{equation*}
For mesh size $h>0$ the corresponding discrete spaces are denoted by
\begin{equation*}
\cH^d_h=\ell^2(h\bZ^d)\otimes \bC^{\nu(d)},
\quad d=1,2,3.
\end{equation*}
The norm on $\cH_h^d$ is given by
\begin{equation*}
\norm{u_h}_{\cH_h^d}^2=h^d\sum_{k\in\bZ^d}
\abs{u_h(k)}^2,\quad u_h\in\cH_h^d.
\end{equation*}
Here $\abs{\,\cdot\,}$ denotes the Euclidean  norm on $\bC^{\nu(d)}$. We index $u_h$ by $k\in\bZ^d$; the $h$ dependence is in the subscript of $u_h$. 

To relate the spaces $\cH_h^d$ and $\cH^d$ we introduce embedding operators $J_h\colon\cH_h^d \to \cH^d$ and discretization operators
$K_h\colon \cH^d \to \cH_h^d$,  constructed from a pair of biorthogonal Riesz sequences, as in~\cite[section~2]{CGJ}. We describe the construction briefly, with further details and assumptions given in section~\ref{sectPrelim}. Let $\vp_0,\psi_0\in L^2(\bR^d)$ and assume that $\{\vp_0(\,\cdot\,-k)\}_{k\in\bZ^d}$ and  $\{\psi_0(\,\cdot\,-k)\}_{k\in\bZ^d}$ are a pair of biorthogonal Riesz sequences in $L^2(\bR^d)$. Define $\vp_{h,k}(x)=\vp_0((x-hk)/h)$, 
and $\psi_{h,k}(x)=\psi_0((x-hk)/h)$, 
$x\in\bR^d$, $k\in\bZ^d$, $h>0$. The embedding operator $J_h$ is then defined as
\begin{equation*}
(J_hu_h)(x)=\sum_{k\in\bZ^d}\vp_{h,k}(x)u_h(k).
\end{equation*}
Note that here $\vp_{h,k}(x)$ is a scalar multiplying a vector $u_h(k)\in\bC^{\nu(d)}$. To construct the discretization operator, let $\widetilde{J}_h$ be defined as $J_h$ with $\vp_0$ replaced by $\psi_0$.
The discretization operator is then defined as $K_h=(\widetilde{J}_h)^{\ast}$. For $d=1,2$, it can be written explicitly as
\begin{equation*}
(K_hf)(k)=\frac{1}{h^d}\begin{bmatrix}
\ip{\psi_{h,k}}{f^1}\\ \ip{\psi_{h,k}}{f^2}
\end{bmatrix},\quad
f=\begin{bmatrix}f^1\\ f^2\end{bmatrix}\in\cH^d.
\end{equation*}
A similar formula holds for $d=3$. We have $K_hJ_h=\identh$, where $\identh$ is the identity in $\cH_h^d$, and $J_hK_h$ is a projection in $\cH^d$ onto $J_h\cH_h^d$.

Let $H_0$ be the free Dirac operator in $\cH^d$, $d=1,2,3$, and let $H_{0,h}$ be an approximation defined on $\cH^d_h$. We compare the operators
\begin{equation*}
J_h(H_{0,h}-z\identh)^{-1}K_h\quad \text{and} \quad (H_0-z\ident)^{-1}
\end{equation*}
acting on $\cH^d$. The question of interest is in what sense will 
$J_h(H_{0,h}-z\identh)^{-1}K_h$ converge to $(H_0-z\ident)^{-1}$ as $h\to0$. We now summarize the results obtained. First we briefly define the operators considered.

Let $\sigma_j$, $j=1,2,3$, denote the Pauli matrices 
\begin{equation}\label{Pauli}
	\sigma_1=\begin{bmatrix} 0 & 1 \\ 1 & 0\end{bmatrix},\quad
	\sigma_2=\begin{bmatrix} 0 & -\I \\ \I & 0\end{bmatrix},\quad
	\sigma_3=\begin{bmatrix} 1 & 0 \\ 0 & -1\end{bmatrix}.
\end{equation}
Let $m\geq0$ denote the mass. To simplify we do not indicate dependence on the mass in the notation for operators.
In dimension $d=1$ the free Dirac operator is given by the operator matrix
\begin{equation*}
H_0=-\I \dfrac{\di}{\di x}\sigma_1+m\sigma_3
\end{equation*}
on $\cH^1$. We consider two discrete approximations based on replacing $-\I \frac{\di}{\di x}$ by finite difference operators. Let $\identhsf$ denote the identity operator on $\ell^2(h\bZ)$. We define
\begin{equation*}
H_{0,h}^{\rm fb}=\begin{bmatrix}
m\identhsf & D^-_h\\ D^+_h & -m\identhsf
\end{bmatrix}
\quad \text{and} \quad
H_{0,h}^{\rm s}=\begin{bmatrix}
m\identhsf & D_h^{\rm s}\\
D_h^{\rm s} & -m\identhsf
\end{bmatrix}.
\end{equation*}
Here the 
forward and backward finite difference operators are defined as
\begin{equation}\label{1Dfb}
(D^+_hu_h)(k)=\frac{1}{\I h}(u_h(k+1)-u_h(k)),\quad
(D^-_hu_h)(k)=\frac{1}{\I h}(u_h(k)-u_h(k-1)),
\end{equation}
and satisfies $(D_h^+)^{\ast}=D_h^-$. The symmetric difference operator is the self-adjoint operator $D_h^{\rm s}=\frac12(D^+_h+D^-_h)$, i.e.
\begin{equation}\label{1Ds}
(D_h^{\rm s}u_h)(k)=\frac{1}{2\I h}(u_h(k+1)-u_h(k-1)).
\end{equation}

In dimension $d=2$ the free Dirac operator is defined as
\begin{equation*}
H_0=-\I\frac{\partial}{\partial x_1}\sigma_1
-\I\frac{\partial}{\partial x_2}\sigma_2 +m\sigma_3
\end{equation*}
on $\cH^2$. As in the $d=1$ case, there are two natural discrete models given by
\begin{equation*}
H^{\rm fb}_{0,h}=
\begin{bmatrix}
m\identhsf & D^-_{h;1}-\I D^-_{h;2}\\
D^+_{h;1}+\I D^+_{h;2} &-m\identhsf
\end{bmatrix}
\end{equation*}
and
\begin{equation*}
H_{0,h}^{\rm s}=\begin{bmatrix}
m\identhsf & D^{\rm s}_{h;1}-\I D^{\rm s}_{h;2}\\
D^{\rm s}_{h;1}+\I D^{\rm s}_{h;2} &-m\identhsf
\end{bmatrix}.
\end{equation*}
Here $D^{\pm}_{h;j}$ and $D^{\rm s}_{h;j}$ are the corresponding finite differences in the $j$'th coordinate.
It turns out that these two discrete models \emph{do not} lead to norm resolvent convergence, so we also define two modified versions. Let $-\Delta_h$ denote the discrete Laplacian; see~\eqref{Lap}. Then the modified operators are given by
\begin{equation*}
\widetilde{H}_{0,h}^{\rm fb}=H^{\rm fb}_{0,h}-h\Delta_h\sigma_3
\quad \text{and} \quad
\widetilde{H}_{0,h}^{\rm s}=H^{\rm s}_{0,h}-h\Delta_h\sigma_3.
\end{equation*}
Here $-h\Delta_h\sigma_3$ is understood to be the operator matrix
\begin{equation*}
	\begin{bmatrix} -h\Delta_h & 0 \\ 0 & h\Delta_h\end{bmatrix}.
\end{equation*}
The details on the discretizations in dimension $d=3$ can be found in section~\ref{sect3D}.
 
Let $\cK_1$ and $\cK_2$ be two Hilbert spaces. The space of bounded operators from $\cK_1$ to $\cK_2$ is denoted by $\cB(\cK_1,\cK_2)$. If $\cK_1=\cK_2=\cK$ we write $\cB(\cK)=\cB(\cK,\cK)$. In the following theorem we collect the positive results obtained on norm resolvent convergence in $\cB(\cH^d)$. We use the convention  $(-0,0)=\emptyset$ in the statements of results.

\begin{theorem}\label{thm11}
Let $H_{0,h}$ be equal to $H_{0,h}^{\rm fb}$, $d=1$, or equal to 
 $\widetilde{H}_{0,h}^{\rm fb}$, $d=2,3$, or equal to
$\widetilde{H}_{0,h}^{\rm s}$,
 $d=1,2,3$. Let $H_0$ denote the free Dirac operator in the corresponding dimension. Then the following result holds.

Let $K\subset (\bC\setminus \bR)\cup (-m,m)$ be compact. Then there exists $C>0$ such that
\begin{equation}\label{general-est}
\norm{
J_h(H_{0,h}-z\identh)^{-1}K_h-(H_0-z\ident)^{-1}
}_{\cB(\cH^d)} \leq C h
\end{equation}
for all $z\in K$ and $h\in(0,1]$.
\end{theorem}

Theorem~\ref{thm11} can be generalized to also include a potential, by following the approach in \cite{CGJ}. Let $V\colon \bR^d\to \cB(\bC^{\nu(d)})$ be bounded and Hölder continuous. Assume $V(x)$ is self-adjoint for each $x\in\bR^d$. Define the discretization as $V_h(k)=V(hk)$ for $k\in\bZ^d$. Then we can define self-adjoint operators $H=H_0+V$ on $\cH^d$ and $H_h=H_{0,h}+V_h$ on $\cH_h^d$ for all the discrete models. The results in Theorem~\ref{thm11} then generalize to $H$ and $H_h$, with an estimate $Ch^{\theta'}$, where $0<\theta'<1$ depends on the Hölder exponent for $V$; see section~\ref{sectV}.

In the next theorem we summarize some negative results with non-convergence in the $\cB(\cH^d)$-operator norm in part (i), and in part (ii) a result using the Sobolev spaces $H^1(\bR^d)\otimes\bC^{\nu(d)}$ is given. In particular, 
the estimate~\eqref{Sobolev-est} implies strong resolvent convergence in $\cB(\cH^d)$.

\begin{theorem}\label{thm12}
	Let $H_{0,h}$ be equal to $H_{0,h}^{\rm fb}$, $d=2,3$, or equal to $H_{0,h}^{\rm s}$, $d=1,2,3$. Let $H_0$ denote the free Dirac operator in the corresponding dimension. Then the following results hold.
\begin{enumerate}[\rm(i)]
\item Let $z\in (\bC\setminus \bR)\cup (-m,m)$. Then
$J_h(H_{0,h}-z\identh)^{-1}K_h$ \emph{does not converge} to $(H_0-z\ident)^{-1}$ in the operator norm on $\cB(\cH^d)$ as $h\to0$.

\item Let $K\subset (\bC\setminus \bR)\cup (-m,m)$ be compact. Then there exists $C>0$ such that
\begin{equation}\label{Sobolev-est}
\norm{
J_h(H_{0,h}-z\identh)^{-1}K_h-(H_0-z\ident)^{-1}
}_{\cB(H^1(\bR^d)\otimes\bC^{\nu(d)},\cH^d)} \leq C h
\end{equation}
for all $z\in K$ and $h\in(0,1]$.
\end{enumerate}
\end{theorem}

The estimate~\eqref{general-est} implies results on the spectra of the operators $H_{0,h}$ and $H_0$ and their relation, see~\cite[section~5]{CGJ}. Such results are not obtainable from the strong convergence implied by the estimate~\eqref{Sobolev-est}.
Thus we are in the remarkable situation that in dimensions $d=2,3$ we need to modify the natural discretizations in order to obtain spectral information. Furthermore, in dimension $d=1$ to obtain spectral information we must use either the forward-backward discretizations or the modified symmetric discretizations. Moreover, this is relevant for resolving the unwanted \emph{fermion doubling} phenomenon that is present in some discretizations of Dirac operators \cite{CPT}.

Results of the type~\eqref{general-est} were first obtained by
Nakamura and Tadano~\cite{NT} for $H=-\Delta+V$ on $L^2(\bR^d)$ and
$H_h=-\Delta_h+V_h$ on $\ell^2(h\bZ^d)$ for a large class of real
potentials~$V$, including unbounded~$V$. They used special cases of
the $J_h$ and $K_h$ as defined here, i.e.\ the pair of biorthogonal
Riesz sequences is replaced by a single orthonormal sequence. 
Recently their results have been applied to quantum graph 
Hamiltonians~\cite{ENT}.
In~\cite{IJ} the continuum limit is studied for a number of different problems. Here strong resolvent convergence is proved up to the spectrum and scattering results are derived.

In~\cite{CGJ} the authors proved results of the
type~\eqref{general-est} for a class of Fourier multipliers~$H_0$ and
their discretizations~$H_{0,h}$, and obtained results of the
type~\eqref{general-est} for perturbations $H=H_0+V$ and
$H_h=H_{0,h}+V_h$ with a bounded, real-valued, and Hölder continuous
potential. Note that the results in~\cite{CGJ} do not directly apply
to Dirac operators, since the free Dirac operators do not satisfy an
essential symmetry condition~\cite[Assumption~3.1(4)]{CGJ}.
In~\cite{SU} Schmidt and Umeda proved strong resolvent convergence for
Dirac operators in dimension $d=2$ using the discretization
$H_{0,h}^{\rm fb}$. They allow a class of bounded non-self-adjoint
potentials
and also state corresponding results for dimensions $d=1,3$. 

The remainder of this paper is organized as follows.
Section~\ref{sectPrelim} introduces additional notation and operators
used in the paper. Sections~\ref{sect1D}, \ref{sect2D}, and
\ref{sect3D} prove Theorem~\ref{thm11} and Theorem~\ref{thm12}(i) in
the one-, two-, and three-dimensional cases, respectively. Since some
of the arguments are very similar in the different dimensions, we will
give the full details in dimension two, and omit parts of the proofs
in dimensions one and three that are essentially the same verbatim.
Theorem~\ref{thm12}(ii) is proved in section~\ref{sectSobolev}.
Finally we show how a potential $V$ can be added to our results in
section~\ref{sectV}.

\section{Preliminaries}\label{sectPrelim}
In this section we collect a number of definitions and results used in the sequel. 

\subsection{Notation for identity operators}

We use the following notation for identity operators on various spaces: $I$ on~$\cH^d$, $I_h$ on~$\cH^d_h$, $\sfI$ on~$L^2(\bR^d)$, $\sfI_h$ on~$\ell^2(h\bZ^d)$, $\bone$ on~$\bC^2$, and $\sfone$ on~$\bC^4$. In section~\ref{sect3D}, in the definitions of the operator matrices for the free Dirac operator and its discretizations, $\bone$ denotes the identity on $L^2(\bR^3)\otimes\bC^2$ and $\bone_h$ denotes the identity on $\ell^2(h\bZ^3)\otimes\bC^2$.

\subsection{Finite differences}

The forward, backward, and symmetric difference operators on $\cH_h^1$ are defined in~\eqref{1Dfb} and~\eqref{1Ds}. Let $\{e_1,e_2,e_3\}$ be the canonical basis in $\bZ^3$. The forward partial difference operators for mesh size $h$ are defined by
\begin{equation}\label{Dj+}
(D^+_{h;j}u_h)(k)=\frac{1}{\I h}\bigl(u_h(k+e_j)-u_h(k)\bigr),
\quad j=1,2,3,
\end{equation}
and backward partial difference operators by
\begin{equation}\label{Dj-}
(D^-_{h;j}u_h)(k)=\frac{1}{\I h}\bigl(u_h(k)-u_h(k-e_j)\bigr),\quad
j=1,2,3.
\end{equation}
The symmetric difference operators are given by
\begin{equation}\label{Djs}
(D^{\rm s }_{h;j}u_h)(k)=\frac{1}{2\I h}\bigl(u_h(k+e_j)-u_h(k-e_j)\bigr),
\quad
j=1,2,3.
\end{equation}
Note that $(D^+_{h;j})^{\ast}=D^-_{h;j}$ and 
$(D^{\rm s}_{h;j})^{\ast}=D^{\rm s}_{h;j}$.

The discrete Laplacian acting on $\ell^2(h\bZ^d)$ is given by
\begin{equation}\label{Lap}
(-\Delta_hv_h)(k)=\frac{1}{h^2}
\sum_{j=1}^d\bigl(2v_h(k)-v_h(k+e_j)-v_h(k-e_j)\bigr).
\end{equation}

\subsection{Fourier transforms}
 
We use Fourier transforms extensively. They are
 normalized to be unitary.
Write $\widehat{\cH}^d=L^2(\bR^d)\otimes\bC^{\nu(d)}$ and let
 $\cF\colon \cH^d \to \widehat{\cH}^d$ 
be the Fourier transform given by
\begin{equation*}
(\cF f)(\xi)=\frac{1}{(2\pi)^{d/2}}\int_{\bR^d}\e^{-\I x\cdot\xi}f(x)\di x, \quad \xi\in\bR^d,
\end{equation*}
with adjoint $\cF^{\ast}\colon \widehat{\cH}^d\to\cH^d$. We suppress their dependence on $d$ in the notation, as it will be obvious in which dimension they are used.

Let $\bT^d_h=[-\frac{\pi}{h},\frac{\pi}{h}]^d$, $d=1,2,3$, and 
$\widehat{\cH}_h^d=L^2(\bT^d_h)\otimes\bC^{\nu(d)}$.
The discrete Fourier transform 
$\sfF_h\colon \cH_h^d\to\widehat{\cH}_h^d$ and its adjoint 
$\sfF_h^{\ast}\colon\widehat{\cH}_h^d\to\cH_h^d$ are given by
\begin{align*}
(\sfF_h u_h)(\xi)
&=\frac{h^d}{(2\pi)^{d/2}}\sum_{k\in\bZ^d}u_h(k)\e^{-\I hk\cdot\xi},\quad \xi\in\bT^d_h,\\
(\sfF_h^{\ast}g)(k)&=\frac{1}{(2\pi)^{d/2}}\int_{\bT^d_h}\e^{\I hk\cdot\xi}g(\xi)\di\xi,\quad  k\in\bZ^d,
\end{align*} 
for $d=1,2,3$.

\subsection{Embedding and discretization operators}

We describe in some detail how the the embedding and discretization operators in~\cite[section~2]{CGJ} are adapted to the Dirac case.

Let $\cK$ be a Hilbert space. Let $\{u_k\}_{k\in\bZ^d}$ and $\{v_k\}_{k\in\bZ^d}$ be two sequences in $\cK$. They are said to be \emph{biorthogonal} if
\begin{equation*}
\ip{u_k}{v_n}=\delta_{k,n}, \quad k,n\in\bZ^d,
\end{equation*}
where $\delta_{k,n}$ is Kronecker's delta.

A sequence $\{u_k\}_{k\in\bZ^d}$ is called a \emph{Riesz sequence}
if there exist $A>0$ and $B>0$ such that
\begin{equation*}
A\sum_{k\in\bZ^d}\abs{c_k}^2\leq\norm{\sum_{k\in\bZ^d}c_ku_k}^2
\leq B\sum_{k\in\bZ^d}\abs{c_k}^2
\end{equation*}
for all $\{c_k\}_{k\in\bZ^d}\in\ell^2(\bZ^d)$.

\begin{assumption}\label{assumpA1}
Let $d=1,2,$ or $3$. Let $\vp_0,\psi_0\in L^2(\bR^d)$. Define
\begin{equation*}
\vp_{h,k}(x)=\vp_0((x-hk)/h),\quad \psi_{h,k}(x)=\psi_0((x-hk)/h),
\quad h>0, \quad k\in\bZ^d.
\end{equation*}
Assume that $\{\vp_{1,k}\}_{k\in\bZ^d}$ and
$\{\psi_{1,k}\}_{k\in\bZ^d}$ are biorthogonal Riesz sequences in 
$L^2(\bR^d)$.
\end{assumption}

To simplify, we omit the dependence on $d$ in the notation for embedding and discretization operators.
The embedding operators $J_h\colon\cH_h^d\to\cH^d$ are defined by
\begin{equation}\label{Jh-def}
J_hu_h=\sum_{k\in\bZ^d}\vp_{h,k}u_h(k),\quad u_h\in\cH_h^d.
\end{equation}
For $d=1,2$, and $u_h(k)=\begin{bmatrix}u_h^1(k)\\ u_h^2(k)\end{bmatrix}$,  the notation above means
\begin{equation*}
\vp_{h,k}u_h(k)=\begin{bmatrix}u_h^1(k)\vp_{h,k}\\ u_h^2(k)\vp_{h,k}
\end{bmatrix},
\end{equation*}
with an obvious modification in case $d=3$.
As a consequence of the Riesz sequence assumption we get
a uniform bound
\begin{equation*}
\sup_{h>0}\norm{J_h}_{\cB(\cH_h^d,\cH^d)}<\infty.
\end{equation*}
The operators $\widetilde{J}_h$ are defined as above by replacing $\vp_{h,k}$ by $\psi_{h,k}$ in \eqref{Jh-def}. Then the discretization operators are defined as $K_h=(\widetilde{J}_h)^{\ast}$. Explicitly, for $d=1,2$,
\begin{equation*}
(K_hf)(k)=\frac{1}{h^d}\begin{bmatrix}
\ip{\psi_{h,k}}{f^1}
\\
\ip{\psi_{h,k}}{f^2}
\end{bmatrix},\quad k\in\bZ^d,
\end{equation*}
with an obvious modification for $d=3$.
We have the uniform bound
\begin{equation*}
\sup_{h>0}\norm{K_h}_{\cB(\cH^d,\cH_h^d)}<\infty.
\end{equation*}
Biorthogonality implies that
\begin{equation*}
K_hJ_h=\identh
\end{equation*}
and that $J_hK_h$ is a projection onto $J_h\cH_h^d$ in $\cH^d$.
A further assumption on the functions $\vp_0$ and $\psi_0$ is needed.
\begin{assumption}[{\cite[Assumption~2.8]{CGJ}}]\label{assumpA2}
Let $\widehat{\vp}_0,\widehat{\psi}_0\in L^2(\bR^d)$ be essentially bounded and satisfy Assumption~{\rm\ref{assumpA1}}. Assume further that there exists $c_0>0$ such that
\begin{equation*}\label{supp-cond-phi}
\supp(\widehat{\vp}_0)\subseteq [-\tfrac{3\pi}{2},\tfrac{3\pi}{2}]^d
\quad\text{and}\quad \abs{\widehat{\vp}_0(\xi)}\geq c_0,\quad
\xi\in[-\tfrac{\pi}{2},\tfrac{\pi}{2}]^d,
\end{equation*}
and
\begin{equation*}\label{supp-cond-psi}
\supp(\widehat{\psi}_0)\subseteq [-\tfrac{3\pi}{2},\tfrac{3\pi}{2}]^d
\quad\text{and}\quad \abs{\widehat{\psi}_0(\xi)}\geq c_0,\quad
\xi\in[-\tfrac{\pi}{2},\tfrac{\pi}{2}]^d.
\end{equation*}
\end{assumption}
For examples of $\vp_0$ and $\psi_0$ satisfying Assumption~\ref{assumpA2}, see \cite[subsection 2.1]{CGJ}.

\subsection{Two lemmas}

We often use the following elementary result, where the identity matrix is denoted by $\ident$. 

\begin{lemma}\label{lemmaG}
Let $G\in\cB(\bC^n)$ be a self-adjoint $n\times n$ matrix. Then
\begin{align}
\norm{G-\I\ident}_{\cB(\bC^n)}
&=\norm{G^2+\ident}_{\cB(\bC^n)}^{1/2},\label{G+}\\
\norm{(G-\I\ident)^{-1}}_{\cB(\bC^n)}
&=\norm{(G^2+\ident)^{-1}}_{\cB(\bC^n)}^{1/2}.\label{G-}
\end{align}
\end{lemma}
\begin{proof}
It suffices to prove~\eqref{G+}.
We use the $C^{\ast}$-identity in $\cB(\bC^n)$ to get
\begin{equation*}
\norm{G-\I\ident}_{\cB(\bC^n)}^2=
\norm{(G+\I\ident)(G-\I\ident)}_{\cB(\bC^n)}
=\norm{G^2+\ident}_{\cB(\bC^n)}. \qedhere
\end{equation*}
\end{proof}

The following lemma will be used in the proofs related to the non-convergence results; see e.g.~\cite[Theorem~XIII.83]{RS}.

\begin{lemma}\label{lemma21}
Let $d=1$, $2$, or $3$.
Assume that $M_h\colon \bT_h^d\to \cB(\bC^{\nu(d)})$ is a continuous 
matrix-valued function. Let $T_{M_h}$ denote the operator of multiplication by $M_h$,
\begin{equation*}
	T_{M_h} = \int_{\Tdh}^{\oplus} M_h(\xi)\di \xi,
\end{equation*}
on $\widehat{\cH}_h^d \simeq L^2(\Tdh;\bC^{\nu(d)})$. Then
\begin{equation}\label{op-norm}
\norm{T_{M_h}}_{\cB(\widehat{\cH}_h^d)}=\max_{\xi\in\bT_h^d}
\norm{M_h(\xi)}_{\cB(\bC^{\nu(d)})}.
\end{equation}
\end{lemma}

\section{The 1D free Dirac operator}\label{sect1D}
We state and prove results for the 1D Dirac operator.  On $\cH^1$
the one-dimensional free Dirac operator with mass $m\geq0$ is given by the operator matrix
\begin{equation*}\label{defH0}
H_0=-\I \dfrac{\di}{\di x}\sigma_1+m\sigma_3=\begin{bmatrix}
m\sfI & -\I\dfrac{\di}{\di x}\\
-\I\dfrac{\di}{\di x} & -m\sfI
\end{bmatrix},
\end{equation*}
where $\sfI$ denotes the identity operator on $L^2(\bR)$.

\subsection{The 1D forward-backward difference model}
\label{sect31}

Using~\eqref{1Dfb} the forward-backward difference model
 of $H_0$ is defined as
\begin{equation*}\label{defH0h}
H_{0,h}^{\rm fb}=\begin{bmatrix}
m\identhsf & D^-_h\\ D^+_h & -m\identhsf
\end{bmatrix},
\end{equation*}
where $\identhsf$ denotes the identity operator on $\ell^2(h\bZ)$. The operators $H_0$ and $H_{0,h}^{\rm fb}$ are given as multipliers in Fourier space by the functions $G_0$ and $G_{0,h}^{\rm fb}$, respectively, where
\begin{equation}\label{1DG0}
G_0(\xi)=\begin{bmatrix}m & \xi\\ \xi & -m\end{bmatrix}
\end{equation}
and
\begin{equation}\label{1DG0h}
G_{0,h}^{\rm fb}(\xi)=\begin{bmatrix}
m & -\frac{1}{\I h}(\e^{-\I h\xi}-1)\\
\frac{1}{\I h}(\e^{\I h\xi}-1)
 & -m
\end{bmatrix}.
\end{equation}
We define
\begin{equation}\label{1Dg0}
g_0(\xi)=m^2+\xi^2,
\end{equation}
and
\begin{equation}\label{1Dg0h}
g_{0,h}^{\rm fb}(\xi)=m^2+\tfrac{4}{h^2}\sin^2(\tfrac{h}{2}\xi).
\end{equation}
Then
\begin{equation}\label{square}
G_0(\xi)^2=g_0(\xi)\bone
\quad
\text{and}
\quad
G_{0,h}^{\rm fb}(\xi)^2=g_{0,h}^{\rm fb}(\xi)\bone.
\end{equation}

\begin{lemma}\label{lemma31}
Assume $\xi\neq0$. Then we have
\begin{equation}\label{eq38}
\nbc{(G_0(\xi)-\I\bone)^{-1}}\leq\frac{1}{\abs{\xi}}.
\end{equation}
There exists $C>0$ such that for  $h\xi\in[-\tfrac{3\pi}{2},\tfrac{3\pi}{2}]$ we have
\begin{equation}\label{eq39}
\nbc{(G_{0,h}^{\rm fb}(\xi)-\I\bone)^{-1}}\leq \frac{C}{\abs{\xi}}.
\end{equation}
\end{lemma}
\begin{proof}
Using Lemma~\ref{lemmaG} together with~\eqref{1Dg0} and~\eqref{square} we get
\begin{align*}
\nbc{(G_0(\xi)-\I\bone)^{-1}}&=\nbc{(G_0(\xi)^2+\bone)^{-1}}^{1/2}=\frac{1}{(1+m^2+\xi^2)^{1/2}}\leq\frac{1}{\abs{\xi}},
\end{align*}
proving~\eqref{eq38}.

To prove~\eqref{eq39} we use Lemma~\ref{lemmaG},~\eqref{1Dg0h}, and~\eqref{square} to get
\begin{align*}
\nbc{(G_0^{\rm fb}(\xi)-\I\bone)^{-1}}
&=\nbc{(G_0^{\rm fb}(\xi)^2+\bone)^{-1}}^{1/2}=\bigl(1+m^2+\tfrac{4}{h^2}\sin^2(\tfrac{h}{2}\xi)\bigr)^{-1/2}\\
&\leq \bigl(\tfrac{4}{h^2}\sin^2(\tfrac{h}{2}\xi)\bigr)^{-1/2}.
\end{align*}
There exists $c>0$ such that for $\abs{\theta}\leq \tfrac{3\pi}{4}$ we have $\abs{\sin(\theta)}\geq c\abs{\theta}$. For 
$h\xi\in[-\tfrac{3\pi}{2},\tfrac{3\pi}{2}]$ then~\eqref{eq39} follows.
\end{proof}

\begin{lemma}\label{lemma32}
There exists $C>0$ such that
\begin{equation*}
\nbc{(G_0(\xi)-\I\bone)^{-1}-(G_{0,h}^{\rm fb}(\xi)-\I\bone)^{-1}}\leq
Ch
\end{equation*}
for $h\xi\in[-\tfrac{3\pi}{2},\tfrac{3\pi}{2}]$.
\end{lemma}
\begin{proof}
We have
\begin{align*}
(G_0(\xi)-\I\bone)^{-1}-&(G_{0,h}^{\rm fb}(\xi)-\I\bone)^{-1}
\\
&=
(G_0(\xi)-\I\bone)^{-1}\bigl(
G_{0,h}^{\rm fb}(\xi)-G_0(\xi)
\bigr)(G_{0,h}^{\rm fb}(\xi)-\I\bone)^{-1}.
\end{align*}
To estimate the $12$ and $21$ entries in $G_{0,h}^{\rm fb}(\xi)-G_0(\xi)$ we use Taylor's formula:
\begin{equation*}\label{taylor}
\e^{\I h\xi}=1+\I h\xi+(\I h\xi)^2
\int_0^1\e^{\I ht\xi}(1-t)\di t.
\end{equation*}
It follows that the $12$ and $21$ entries  are estimated by $Ch\abs{\xi}^2$. Using Lemma~\ref{lemma31} the result follows.
\end{proof}

Using Lemmas~\ref{lemma31} and~\ref{lemma32} we can adapt the arguments in~\cite{CGJ} to obtain the following result. We omit the details here, and refer the reader to the proof of Theorem~\ref{thm44} where details of the adaptation are given.
\begin{theorem}\label{thm33}
Let $K\subset (\bC\setminus \bR)\cup(-m,m)$ be compact. Then there exists $C>0$ such that
\begin{equation*}
\norm{
J_h(H_{0,h}^{\rm fb}-z\identh)^{-1}K_h-(H_0-z\ident)^{-1}
}_{\cB(\cH^1)} \leq C h
\end{equation*}
for all $z\in K$ and $h\in(0,1]$.
\end{theorem}

\subsection{The 1D symmetric difference model}
\label{sect32}

The discrete model based on the symmetric difference operator~\eqref{1Ds} is
\begin{equation}\label{1DH0hs}
H_{0,h}^{\rm s}=\begin{bmatrix}
m\identhsf & D_h^{\rm s}\\
D_h^{\rm s} & -m\identhsf
\end{bmatrix}.
\end{equation}
In Fourier space it is a multiplier with symbol
\begin{equation}\label{1DG0hs}
G_{0,h}^{\rm s}(\xi)=\begin{bmatrix}
m & \frac{1}{h}\sin(h\xi)\\
\frac{1}{h}\sin(h\xi) & -m
\end{bmatrix}.
\end{equation}
We have
\begin{equation}\label{1Dg0hs}
G_{0,h}^{\rm s}(\xi)^2=g_{0,h}^{\rm s}(\xi)
\bone \quad \text{where} \quad
g_{0,h}^{\rm s}(\xi)=m^2+\tfrac{1}{h^2}\sin^2(h\xi).
\end{equation}

\begin{lemma}\label{lemma34}
There exists $c>0$ such that
\begin{equation}\label{new-bound}
\max_{\xi\in\bT^1_h}
\nbc{(G_{0,h}^{\rm fb}(\xi)-\I\bone)^{-1}-
(G_{0,h}^{\rm s}(\xi)-\I\bone)^{-1}}\geq c
\end{equation}
for all $h\in (0,1]$.
\end{lemma}

\begin{proof}
We have
\begin{align}
\bigl(G_{0,h}^{\rm fb}(\xi)-\I\bone\bigr)
\bigl[
(G_{0,h}^{\rm s}(\xi)-\I\bone)^{-1}&-
(G_{0,h}^{\rm fb}(\xi)-\I\bone)^{-1}
\bigr]
\bigl(G_{0,h}^{\rm s}(\xi)-\I\bone\bigr)
\notag\\
&=G_{0,h}^{\rm fb}(\xi)-G_{0,h}^{\rm s}(\xi)
=\tfrac{1}{h}(1-\cos(h\xi))\sigma_2. \label{eq:lem34eq}
\end{align}
From \eqref{square} and \eqref{1Dg0hs}: $(1+g_{0,h}^{\rm fb}(\xi))^{-1/2}(G_{0,h}^{\rm fb}(\xi)-\I \bone)$ and $(1+g_{0,h}^{\rm s}(\xi))^{-1/2}(G_{0,h}^{\rm s}(\xi)-\I \bone)$ are unitary matrices for all $\xi\in\bT_h^1$. Since $\sigma_2$ is also unitary, \eqref{eq:lem34eq} gives the norm equality
\begin{equation*}
\nbc{(G_{0,h}^{\rm fb}(\xi)-\I\bone)^{-1}-
(G_{0,h}^{\rm s}(\xi)-\I\bone)^{-1}}
= \frac{1-\cos(h\xi)}{h(1+g_{0,h}^{\rm fb}(\xi))^{1/2}(1+g_{0,h}^{\rm s}(\xi))^{1/2}}.
\end{equation*}
If we take $h\xi=\pi$, the right hand side becomes
\begin{equation*}
\frac{2}{\sqrt{(1+m^2)h^2+4}\sqrt{1+m^2}}.
\end{equation*}
Thus for $0<h\leq1$ one can take $c=2((1+m^2)^2+4(1+m^2))^{-1/2}$ in~\eqref{new-bound}.
This concludes the proof.
\end{proof}
Using Lemmas~\ref{lemma21} and~\ref{lemma34} together with Theorem~\ref{thm33} and properties of $J_h$ and $K_h$, we get the following result.
\begin{theorem}
Let $z\in(\bC\setminus\bR)\cup(-m,m)$. Then $J_h(H_{0,h}^{\rm s}-z\identh)^{-1}K_h$ \emph{does not converge} to $(H_0-z\ident)^{-1}$ in the operator norm on $\cB(\cH^1)$ as $h\to 0$.
\end{theorem}

We can introduce a modified operator $\widetilde{H}^{\rm s}_{0,h}$ given by
\begin{equation*}
	\mH^{\rm s}=H_{0,h}^{\rm s}+\begin{bmatrix}
		-h\Delta_h & 0 \\ 0 & h\Delta_h
	\end{bmatrix},
\end{equation*}
where $-\Delta_h$ is the 1D discrete Laplacian; see \eqref{Lap}. We obtain norm resolvent convergence for the modified symmetric difference model, similar to the results in dimensions two and three; see Theorems~\ref{thm44} and~\ref{thm54}. The proof is omitted as it is nearly identical to the proof of Theorem~\ref{thm44}.

\begin{theorem}
	Let $K\subset (\bC\setminus \bR)\cup(-m,m)$ be compact. Then there exists $C>0$ such that
	\begin{equation*}
		\norm{
			J_h(\mH^{\rm s}-z\identh)^{-1}K_h-(H_0-z\ident)^{-1}
		}_{\cB(\cH^1)} \leq C h
	\end{equation*}
	for all $z\in K$ and $h\in(0,1]$.
\end{theorem}

\section{The 2D free Dirac operator}\label{sect2D}
In two dimensions the free Dirac operator on $\cH^2$ with mass $m\geq0$ is given by
\begin{equation}\label{2DH0}
H_0=-\I\frac{\partial}{\partial x_1}\sigma_1
-\I\frac{\partial}{\partial x_2}\sigma_2 +m\sigma_3 = \begin{bmatrix}
	m\sfI & -\I\frac{\partial}{\partial x_1}-\frac{\partial}{\partial x_2}\\
	-\I\frac{\partial}{\partial x_1}+\frac{\partial}{\partial x_2} & -m\sfI
\end{bmatrix},
\end{equation}
where the Pauli matrices are given in~\eqref{Pauli}.
In $\widehat{\cH}^2$ it is a Fourier multiplier with symbol
\begin{equation}\label{2DG0}
G_0(\xi)=\begin{bmatrix}
m & \xi_1-\I \xi_2\\
\xi_1+\I \xi_2 & -m
\end{bmatrix}.
\end{equation}
The corresponding discrete Dirac operator can be obtained by replacing the derivatives in~\eqref{2DH0} by finite differences.

\subsection{The 2D symmetric difference model}
\label{sect41}
We first consider the model 
obtained by using the symmetric difference operators; see~\eqref{Djs} for the definition.
\begin{equation}
H_{0,h}^{\rm s}=\begin{bmatrix}\label{2DH0h}
m\identhsf & D^{\textrm{s}}_{h;1}-\I D^{\textrm{s}}_{h;2}\\
D^{\textrm{s}}_{h;1}+\I D^{\textrm{s}}_{h;2} &-m\identhsf
\end{bmatrix}.
\end{equation}
In $\widehat{\cH}_h^2$ it acts as a Fourier multiplier with symbol
\begin{equation}\label{2DG0h}
G_{0,h}^{\rm s}(\xi)=
\begin{bmatrix}
m & \frac{1}{h}\sin(h\xi_1)- \frac{\I}{h}\sin(h\xi_2)\\
\frac{1}{h}\sin(h\xi_1)+ \frac{\I}{h}\sin(h\xi_2) & -m
\end{bmatrix}.
\end{equation}

The 2D discrete Laplacian is defined in~\eqref{Lap}. We introduce the modified symmetric difference model as
\begin{equation*}\label{2Dmod}
\mH^{\rm s}=H_{0,h}^{\rm s}+\begin{bmatrix}
	-h\Delta_h & 0 \\ 0 & h\Delta_h
\end{bmatrix}.
\end{equation*}
We will show that 
$J_h(\mH^{\rm s}-z\identh)^{-1}K_h$ converges in norm to $(H_0-z\ident)^{-1}$.

In $\widehat{\cH}_h^2$ the operator $\mH^{\rm s}$ acts as a Fourier multiplier with symbol
\begin{equation}\label{GG}
	\widetilde{G}_{0,h}^{\rm s}(\xi)=G_{0,h}^{\rm s}(\xi)+f_h(\xi)\begin{bmatrix}
		1 & 0\\
		0 & -1
	\end{bmatrix}
\end{equation}
where
\begin{equation}\label{fh}
	f_h(\xi)=\tfrac{4}{h}\sin^2(\tfrac{h}{2}\xi_1) + \tfrac{4}{h}\sin^2(\tfrac{h}{2}\xi_2).
\end{equation}
Related to the symbols $G_0$, $G_{0,h}^{\rm s}$, and $\widetilde{G}_{0,h}^{\rm s}$, we define
\begin{equation}\label{g0}
g_0(\xi)=m^2+\xi_1^2+\xi_2^2,
\end{equation}
\begin{equation}\label{g}
g_{0,h}^{\rm s}(\xi)=m^2+\tfrac{1}{h^2}\sin^2(h\xi_1)+\tfrac{1}{h^2}\sin^2(h\xi_2),
\end{equation}
and
\begin{equation}\label{gtilde}
\tilde{g}_{0,h}^{\rm s}(\xi)=\bigl(
m+f_h(\xi)
\bigr)^2+ \tfrac{1}{h^2}\sin^2(h\xi_1)+\tfrac{1}{h^2}\sin^2(h\xi_2).
\end{equation}
We have
\begin{equation} \label{eq:gsquared2D}
	G_0(\xi)^2=g_0(\xi)\bone, \quad G_{0,h}^{\rm s}(\xi)^2= 
	g_{0,h}^{\rm s}(\xi)\bone, \quad \text{and} \quad \widetilde{G}_{0,h}^{\rm s}(\xi)^2= 
	\tilde{g}_{0,h}^{\rm s}(\xi)\bone.
\end{equation}

\begin{lemma}\label{lemma42}
For $\xi\neq0$ we have
\begin{equation}\label{Gest}
\nbc{(G_0(\xi)-\I\bone)^{-1}}\leq\frac{1}{\abs{\xi}}.
\end{equation}
There exists $C>0$ such that for $h\xi\in[-\tfrac{3\pi}{2},\tfrac{3\pi}{2}]^2$ we have
\begin{equation}\label{Gtildeest}
\nbc{
(\widetilde{G}_{0,h}^{\rm s}(\xi)-\I\bone)^{-1}}
\leq\frac{C}{\abs{\xi}}.
\end{equation}
\end{lemma}
\begin{proof}
Lemma~\ref{lemmaG} and \eqref{eq:gsquared2D} imply
\begin{equation*}
\nbc{(G_0(\xi)-\I\bone)^{-1}} = \frac{1}{(1+g_0(\xi))^{1/2}}\leq \frac{1}{\abs{\xi}},
\end{equation*}
such that~\eqref{Gest} holds.

To prove~\eqref{Gtildeest} we first use 
Lemma~\ref{lemmaG} and \eqref{eq:gsquared2D} to get
\begin{equation*}
\nbc{
(\widetilde{G}_{0,h}^{\rm s}(\xi)-\I\bone)^{-1}}
=\frac{1}{(1+\tilde{g}_{0,h}^{\rm s}(\xi))^{1/2}}.
\end{equation*}
Then
note that there exists $c>0$ such that $\abs{\sin(\theta)}\geq c \abs{\theta}$ for $\theta\in[-\tfrac{3\pi}{4},\tfrac{3\pi}{4}]$. 
Thus for $\abs{h\xi_j}\leq \tfrac{3\pi}{4}$, $j=1,2$, we have
\begin{equation*}
\tfrac{1}{h^2}\sin^2(h\xi_j)\geq c_1 \abs{\xi_j}^2,\quad j=1,2.
\end{equation*}
For $\tfrac{3\pi}{4}\leq\abs{h\xi_j}\leq \tfrac{3\pi}{2}$ we have
\begin{equation*}
\tfrac{1}{h}\sin^2(\tfrac{h}{2}\xi_j)\geq c_1 h\abs{\xi_j}^2\geq c_2\abs{\xi_j},
\quad j=1,2.
\end{equation*}
Combining these estimates we get
\begin{equation*}
\tilde{g}_{0,h}^{\rm s}(\xi)\geq c\abs{\xi}^2,
\quad h\xi\in[-\tfrac{3\pi}{2},\tfrac{3\pi}{2}]^2.
\end{equation*}
The estimate~\eqref{Gtildeest} follows.
\end{proof}

\begin{lemma}\label{lemma43}
There exists $C>0$ such that
\begin{equation*}
\nbc{(G_0(\xi)-\I\bone)^{-1}-(\widetilde{G}_{0,h}^{\rm s}(\xi)-\I\bone)^{-1}}\leq
Ch
\end{equation*}
for $h\xi\in[-\tfrac{3\pi}{2},\tfrac{3\pi}{2}]^2$.
\end{lemma}
\begin{proof}
We have
\begin{equation*}
(G_0(\xi)-\I\bone)^{-1}
-(\widetilde{G}_{0,h}^{\rm s}(\xi)-\I\bone)^{-1}=
(G_0(\xi)-\I\bone)^{-1}\bigl(
\widetilde{G}_{0,h}^{\rm s}(\xi)-G_0(\xi)
\bigr)(\widetilde{G}_{0,h}^{\rm s}(\xi)-\I\bone)^{-1}.
\end{equation*}
The $11$ entry in $\widetilde{G}_{0,h}^{\rm s}(\xi)-G_0(\xi)$ is estimated using $\abs{\sin(\theta)}\leq\abs{\theta}$. We get
\begin{equation}\label{11-entry}
\abs{[\widetilde{G}_{0,h}^{\rm s}(\xi)-G_0(\xi)]_{11}}
\leq\abs{f_h(\xi)}\leq C h\abs{\xi}^2.
\end{equation}
The same estimate holds for the $22$ entry.

Taylor's formula yields 
\begin{equation}\label{taylor-sin}
\sin(\theta)=\theta-\frac12 \theta^3\int_0^1\cos(t\theta)(1-t)^2\di t.
\end{equation}
This result implies the estimates
\begin{equation*}
\abs{\tfrac{1}{h}\sin(h\xi_j)-\xi_j}\leq C h^2\abs{\xi_j}^3,\quad
j=1,2,
\end{equation*}
that are used to estimate the $12$ and $21$ entries in
$\widetilde{G}_{0,h}^{\rm s}(\xi)-G_0(\xi)$.

Combining these results with the estimates from Lemma~\ref{lemma42} we get
\begin{equation*}
\nbc{(G_0(\xi)-\I\bone)^{-1}-(\widetilde{G}_{0,h}^{\rm s}(\xi)-\I\bone)^{-1}}\leq
C\frac{h\abs{\xi}^2+h^2\abs{\xi}^3}{\abs{\xi}^2}
=C(h+h^2\abs{\xi})\leq Ch,
\end{equation*}
for $h\xi\in[-\tfrac{3\pi}{2},\tfrac{3\pi}{2}]^2$.
\end{proof} 
We state~\cite[Lemma~3.3]{CGJ} in a form adapted to the Dirac operators and outline its proof. 
\begin{lemma}\label{lemmaA3}
Let $d=1,2$, or $3$. Let $H_0$ be the free Dirac operator in $\cH^d$. 
Let $\vp_0$ and $\psi_0$ satisfy Assumption~{\rm\ref{assumpA2}}. 
Let $K\subset(\bC\setminus\bR)\cup(-m,m)$ be compact. Then there exists $C>0$ such that
\begin{equation*}\label{eqA9}
\norm{(J_hK_h-\ident)(H_0-z\ident)^{-1}}_{\cB(\cH^d)}\leq C h,
\end{equation*}
for all $z\in K$ and $h\in(0,1]$.
\end{lemma}

\begin{proof}
We assume $d=2$. It suffices to consider $K=\{\I\}$, since $(H_0-\I\ident)(H_0-z\ident)^{-1}$ is bounded uniformly in norm for $z\in K$.
Let $u\in\cS(\bR^2)\otimes\bC^2$, the Schwartz space. Going through the computations in~\cite[section~2]{CGJ} using that $\vp_0$ and $\psi_0$ are scalar functions, we get the result
\begin{align*}
\bigl(\cF(J_hK_h-\ident)&(H_0-\I\ident)^{-1}\cF^{\ast}u\bigr)(\xi) &
\\
&=
(2\pi)^d\widehat{\vp}_0(h\xi)
\sum_{j\in\bZ^2}\overline{\widehat{\psi}_0(h\xi+2\pi j)}(G_0(\xi+\tfrac{2\pi}{h}j)-\I\bone)^{-1}u(\xi+\tfrac{2\pi}{h}j)\\
&\phantom{{}={}} - (G_0(\xi)-\I\bone)^{-1}u(\xi), \quad \xi\in\bR^2.
\end{align*}
Here $G_0$ is given by~\eqref{2DG0}. If $h\xi\in[-\tfrac{\pi}{2},\tfrac{\pi}{2}]^2$ then the $j=0$ term is the only non-zero term in the sum. Using~\cite[Lemma~2.7]{CGJ} we conclude that this term and the last term cancel. For $h\xi\notin[-\tfrac{\pi}{2},\tfrac{\pi}{2}]^2$ we use Lemma~\ref{lemma42} to get $\nbc{(G_0(\xi)-\I\bone)^{-1}}\leq C h$, $0<h\leq1$. Since $\widehat{\vp}_0$ and $\widehat{\psi}_0$ are assumed essentially bounded, we conclude that the $j=0$ term in the sum and the last term are bounded by 
$C h\norm{u}_{\widehat{\cH}^2}$. 

Due to the support assumptions on $\widehat{\vp}_0$ and $\widehat{\psi}_0$, only the terms in the sum with $\abs{j}\leq1$ contribute. Assume $\abs{j}=1$ and $h\xi\in \supp(\widehat{\vp}_0)\cap
\supp(\widehat{\psi}_0(\cdot+2\pi j))$. Then for some $c_0>0$ we have $\abs{\xi+\frac{2\pi}{h}j}\geq \frac{c_0}{h}$, which by Lemma~\ref{lemma42} implies
\begin{equation*}
\nbc{(G_0(\xi+\tfrac{2\pi}{h}j)-\I\bone)^{-1}}\leq C h.
\end{equation*}
Again using the boundedness of $\widehat{\vp}_0$ and $\widehat{\psi}_0$ we conclude that
\begin{equation*}
\bigl\lvert(2\pi)^d\widehat{\vp}_0(h\xi)
\overline{\widehat{\psi}_0(h\xi+2\pi j)}(G_0(\xi+\tfrac{2\pi}{h}j)-\I\bone)^{-1}u(\xi+\tfrac{2\pi}{h}j)
\bigr\rvert
\leq C h\abs{u(\xi+\tfrac{2\pi}{h}j)},
\end{equation*}
$0<h\leq1$, $\xi\in\bR^2$. Squaring and integrating the result gives an estimate of the form $Ch\norm{u}_{\widehat{\cH}^2}$. By density, adding up the finite number of terms corresponding to $\abs{j}\leq 1$ gives the final result.
\end{proof}

We have now established the estimates necessary to repeat the arguments from~\cite{CGJ}.
Using the embedding operators $J_h$ and discretization operators $K_h$ defined in section~\ref{sectPrelim}, we state the result and then show in some detail how the arguments in~\cite{CGJ} are adapted to the Dirac case.

\begin{theorem}\label{thm44}
Let $K\subset (\bC\setminus \bR)\cup(-m,m)$ be compact. Then there exists $C>0$ such that
\begin{equation*}
\norm{
J_h(\mH^{\rm s}-z\identh)^{-1}K_h-(H_0-z\ident)^{-1}
}_{\cB(\cH^2)} \leq C h
\end{equation*}
for all $z\in K$ and $h\in(0,1]$.
\end{theorem}
\begin{proof}
We start by proving the result for $K=\{\I\}$. We have
\begin{align*}
J_h (\mH^{\rm s}-\I\identh)^{-1}K_h-
(H_0-\I\ident)^{-1}
&= 
J_h(\mH^{\rm s}-\I\identh)^{-1}K_h-
J_hK_h (H_0-\I\ident)^{-1}
\\
&\phantom{{}={}}+(J_hK_h-\ident)(H_0-\I\ident)^{-1}. 
\end{align*}
The last term is estimated using Lemma~\ref{lemmaA3}.

To estimate the remaining terms we go to Fourier space. 
We have
\begin{align}
\cF\bigl(J_h(\mH^{\rm s}-\I\identh)^{-1}K_h-
J_hK_h(H_0-\I\ident)^{-1}\bigr)\cF^{\ast} 
&= \cF J_h\sfF_h^{\ast}\sfF_h(\mH^{\rm s}-\I\identh)^{-1}
\sfF_h^{\ast}\sfF_h
K_h\cF^{\ast}
\notag
\\
&\phantom{{}={}}
-\cF J_hK_h\cF^{\ast}\cF (H_0-\I\ident)^{-1}\cF^{\ast}.
\label{rewrite1}
\end{align}
Let $u\in\cS(\bR^2)\otimes\bC^2$. We now use a modified version of the computation leading to~\cite[equation~(2.11)]{CGJ}. For the first term we get
\begin{multline}
\bigl[\cF J_h\sfF_h^{\ast}\sfF_h
(\mH^{\rm s}-\I\identh)^{-1}\sfF_h^{\ast}\sfF_h
K_h\cF^{\ast}u\bigr](\xi)\\
=(2\pi)^d\widehat{\vp}_0(h\xi)
(\widetilde{G}_{0,h}^{\rm s}(\xi)-\I\bone)^{-1}\sum_{j\in\bZ^2}
\overline{\widehat{\psi}_0(h\xi+2\pi j)}
u(\xi+\tfrac{2\pi}{h}j).\label{term1}
\end{multline}
For the second term we get
\begin{multline}
\bigl[\cF J_hK_h\cF^{\ast} 
\cF (H_0-\I\ident)^{-1}\cF^{\ast}u\bigr](\xi)
\\
=(2\pi)^d\widehat{\vp}_0(h\xi)
\sum_{j\in\bZ^2}
\overline{\widehat{\psi}_0(h\xi+2\pi j)}
(G_0(\xi+\tfrac{2\pi}{h}j)-\I\bone)^{-1}
u(\xi+\tfrac{2\pi}{h}j).\label{term2}
\end{multline}
We need to rewrite~\eqref{term1}. First we note that 
\begin{equation*}
\widetilde{G}_{0,h}^{\rm s}(\xi)
=\widetilde{G}_{0,h}^{\rm s}(\xi+\tfrac{2\pi}{h}j),
\quad j\in\bZ^2.
\end{equation*}
Next we can rewrite part of~\eqref{term1} as follows, since $\widehat{\psi}_0$ is a scalar-valued function:
\begin{multline*}
(\widetilde{G}_{0,h}^{\rm s}(\xi)-\I\bone)^{-1}\sum_{j\in\bZ^2}
\overline{\widehat{\psi}_0(h\xi+2\pi j)}
u(\xi+\tfrac{2\pi}{h}j)
\\
=\sum_{j\in\bZ^2}
\overline{\widehat{\psi}_0(h\xi+2\pi j)}
(\widetilde{G}_{0,h}^{\rm s}(\xi+\tfrac{2\pi}{h}j)-\I\bone)^{-1}
u(\xi+\tfrac{2\pi}{h}j).
\end{multline*}
We now insert~\eqref{term2} and the rewritten~\eqref{term1}
into~\eqref{rewrite1} to get
\begin{equation*}
\cF\bigl(
J_h(\mH^{\rm s}-\I\identh)^{-1}K_h-J_hK_h (H_0-\I\ident)^{-1}
\bigr)\cF^{\ast}u=\sum_{j\in\bZ^2}q_{j,h},
\end{equation*}
where
\begin{multline*}
q_{j,h}(\xi)=
(2\pi)^d\widehat{\vp}_0(h\xi)\overline{\widehat{\psi}_0(h\xi+2\pi j)}
\\
\times\bigl(
(\widetilde{G}_{0,h}^{\rm s}(\xi+\tfrac{2\pi}{h}j)-\I\bone)^{-1}
-(G_0(\xi+\tfrac{2\pi}{h}j)-\I\bone)^{-1}
\bigr)u(\xi+\tfrac{2\pi}{h}j).
\end{multline*}
Due to the support conditions on $\widehat{\vp}_0$ and $\widehat{\psi}_0$ in Assumption~\ref{assumpA2}, only terms with $\abs{j}\leq1$ contribute. First consider $j=0$. We have assumed
$\supp(\widehat{\vp}_0)$, $\supp(\widehat{\psi}_0)\subseteq [-\tfrac{3\pi}{2},\tfrac{3\pi}{2}]^2$. Using Lemma~\ref{lemma43} and Assumption~\ref{assumpA2} we get
\begin{equation*}
\norm{q_{0,h}}_{\widehat{\cH}^2}\leq C h 
\norm{u}_{\widehat{\cH}^2}.
\end{equation*}
Fix $j\in\bZ^2$ with $\abs{j}=1$. Define
\begin{equation*}
	M=\supp(\widehat{\vp}_0)\cap\supp(\widehat{\psi}_0(\,\cdot+2\pi j)).
\end{equation*}
From the supports of $\widehat{\vp}_0$ and $\widehat{\psi}_0$ we have
\begin{equation*}
	M\subseteq\{\zeta\in[-\tfrac{3\pi}{2},\tfrac{3\pi}{2}]^2 : \abs{\zeta+2\pi j}\geq \tfrac{\pi}{2}\}.
\end{equation*}
Assume $h\xi\in M$, then Lemma~\ref{lemma42} implies
\begin{equation*}
\nbc{(G_0(\xi+\tfrac{2\pi}{h}j)-\I\bone)^{-1}}\leq C h,
\end{equation*}
and

\begin{equation*}
\nbc{(\widetilde{G}_{0,h}^{\rm s}(\xi+\tfrac{2\pi}{h}j)
-\I\bone)^{-1}}\leq C h.
\end{equation*}
These estimates imply
\begin{equation*}
\norm{q_{j,h}}_{\widehat{\cH}^2}\leq C h 
\norm{u}_{\widehat{\cH}^2},\quad \abs{j}=1.
\end{equation*}
Since we have a finite number of $j$ with $\abs{j}\leq1$ and since $u$ is in a dense set, the estimate in Theorem~\ref{thm44} follows in the $K=\{\I\}$ case. For the general case we use the estimates
\begin{equation*}
\nbc{(G_0(\xi)-z\bone)(G_0(\xi)-\I\bone)^{-1}}\leq C
\end{equation*}
and
\begin{equation*}
\nbc{(\widetilde{G}_{0,h}^{\rm s}(\xi)-z\bone)
(\widetilde{G}_{0,h}^{\rm s}(\xi)-\I\bone)^{-1}}\leq C
\end{equation*}
for $z\in K$ where $K\subset(\bC\setminus\bR)\cup(-m,m)$ is compact. Combining these estimates with Lemma~\ref{lemma43} we get for $h\xi\in[-\tfrac{3\pi}{2},\tfrac{3\pi}{2}]^2$
\begin{equation*}
\nbc{(\widetilde{G}_{0,h}^{\rm s}(\xi)-z\bone)^{-1}
-(G_0(\xi)-z\bone)^{-1}}\leq C h,\quad z\in K.
\end{equation*}
This is the crucial estimate used above. Further details are omitted.
\end{proof}

Next we show that, without modification to the symmetric difference model, the norm convergence stated in the theorem fails.

\begin{lemma}\label{lemma45}
There exists $c>0$ such that 
\begin{equation}\label{LB}
\max_{\xi\in\bT_h^2}\norm{(G_{0,h}^{\rm s}(\xi)-\I\bone)^{-1}
-(\widetilde{G}_{0,h}^{\rm s}(\xi)-\I\bone)^{-1}}_{\cB(\bC^2)}\geq c
\end{equation}
for all $h\in(0,1]$.
\end{lemma}
\begin{proof}
Using the notation from~\eqref{GG} we have
\begin{align*}
(G_{0,h}^{\rm s}(\xi)-\I\bone)
\bigl[
(G_{0,h}^{\rm s}(\xi)-\I\bone)^{-1}
&
-(\widetilde{G}_{0,h}^{\rm s}(\xi)-\I\bone)^{-1}
\bigr]
(\widetilde{G}_{0,h}^{\rm s}(\xi)-\I\bone)
\\
&=\widetilde{G}_{0,h}^{\rm s}(\xi) - G_{0,h}^{\rm s}(\xi)
=f_h(\xi)\sigma_3.
\end{align*}
From the same reasoning as in the proof of Lemma~\ref{lemma34}, we obtain
\begin{equation}
\label{lbound}
\nbc{
(G_{0,h}^{\rm s}(\xi)-\I\bone)^{-1}
-(\widetilde{G}_{0,h}^{\rm s}(\xi)-\I\bone)^{-1}
}
=\frac{f_h(\xi)}{(1+g_{0,h}^{\rm s}(\xi))^{1/2}
(1+\tilde{g}_{0,h}^{\rm s}(\xi))^{1/2}}.
\end{equation}
Here $g_{0,h}^{\rm s}(\xi)$ is given by~\eqref{g},
$\tilde{g}_{0,h}^{\rm s}(\xi)$ by~\eqref{gtilde}, and $f_h(\xi)$ by~\eqref{fh}. 

Take $h\xi_1=\pi$ and $h\xi_2=\pi$, and insert them in the last term in~\eqref{lbound}. We get
\begin{equation*}
\max_{\xi\in\bT_h^2}\norm{(G_{0,h}^{\rm s}(\xi)-\I\bone)^{-1}
	-(\widetilde{G}_{0,h}^{\rm s}(\xi)-\I\bone)^{-1}}_{\cB(\bC^2)}\geq \frac{8}{\sqrt{1+m^2}\sqrt{h^2+(8+hm)^2}}.
\end{equation*}
The result~\eqref{LB} then holds for $0<h\leq1$ with 
$c=8[(1+m^2)(1+(8+m)^2)]^{-1/2}$.
\end{proof}

Combining Theorem~\ref{thm44} with Lemma~\ref{lemma21} and 
Lemma~\ref{lemma45}, we obtain the following result using the operators $J_h$ and $K_h$ introduced in section~\ref{sectPrelim}.
\begin{theorem}\label{thm46}
Let $z\in(\bC\setminus\bR)\cup(-m,m)$. Then $J_h(H_{0,h}^{\rm s}-z\identh)^{-1}K_h$ \emph{does not converge} to $(H_0-z\ident)^{-1}$ in the operator norm on $\cB(\cH^2)$ as $h\to 0$.
\end{theorem}

\subsection{The 2D forward-backward difference model}\label{sect2Dfb}
We now consider the model for the discrete Dirac operator obtained by using the forward and backward difference operators; see~\eqref{Dj+} and~\eqref{Dj-} for definitions. The discretized operator is given by
\begin{equation}\label{2DH0fb}
H^{\rm fb}_{0,h}=
\begin{bmatrix}
m\identhsf & D^-_{h;1}-\I D^-_{h;2}\\
D^+_{h;1}+\I D^+_{h;2} &-m\identhsf
\end{bmatrix}.
\end{equation}
In $\cH_h^2$ it is a Fourier multiplier with the symbol
\begin{equation*}
\Gfb(\xi)=
\begin{bmatrix}
m & -\frac{1}{\I h}(\e^{-\I h\xi_1}-1)
+\frac{1}{h}(\e^{-\I h\xi_2}-1)\\
\frac{1}{\I h}(\e^{\I h\xi_1}-1)
+\frac{1}{h}(\e^{\I h\xi_2}-1) & -m
\end{bmatrix}.
\end{equation*}
We also consider the modified model, where the modification is the same as in the symmetric case, i.e.
\begin{equation*}\label{2Dmodfb}
\tHfb=\Hfb+\begin{bmatrix}-h\Delta_h & 0 \\ 0 & h\Delta_h\end{bmatrix}.
\end{equation*}
The corresponding Fourier multiplier is
\begin{equation*}
\tGfb(\xi)=\Gfb(\xi)+
f_h(\xi)\begin{bmatrix} 1 & 0 \\ 0 & -1\end{bmatrix},
\end{equation*}
where $f_h(\xi)$ is given by~\eqref{fh}. We recall the expression
\begin{equation*}
f_h(\xi)=\tfrac{4}{h}\sin^2(\tfrac{h}{2}\xi_1) + \tfrac{4}{h}\sin^2(\tfrac{h}{2}\xi_2).
\end{equation*}
Define
\begin{align}
\gfb(\xi)&=m^2+\tfrac{4}{h^2}\sin^2(\tfrac{h}{2}\xi_1)+\tfrac{4}{h^2}\sin^2(\tfrac{h}{2}\xi_2)\notag\\
&\phantom{{}={}} + \tfrac{2}{h^2}\sin(h(\xi_1-\xi_2)) - \tfrac{2}{h^2}\sin(h\xi_1)
 + \tfrac{2}{h^2}\sin(h\xi_2)
\label{gfb}
\end{align}
and
\begin{align}
\tgfb(\xi)&=\bigl(m+f_h(\xi)
\bigr)^2 + \tfrac{4}{h^2}\sin^2(\tfrac{h}{2}\xi_1) + \tfrac{4}{h^2}\sin^2(\tfrac{h}{2}\xi_2)\notag\\
&\phantom{{}={}} + \tfrac{2}{h^2}\sin(h(\xi_1-\xi_2)) - \tfrac{2}{h^2}\sin(h\xi_1)
 + \tfrac{2}{h^2}\sin(h\xi_2).
\label{tgfb}
\end{align}
Straightforward computations show that
\begin{equation*}
\Gfb(\xi)^2=\gfb(\xi)\bone
\quad\text{and}\quad
\tGfb(\xi)^2=\tgfb(\xi)\bone.
\end{equation*}
We now prove the analogue of~\eqref{Gtildeest} for $\tGfb(\xi)$.

\begin{lemma} \label{lemma49}
There exists $C>0$ such that for $h\xi\in[-\tfrac{3\pi}{2},\tfrac{3\pi}{2}]^2$ we have
\begin{equation*}\label{tGfb-est}
\nbc{(\tGfb(\xi)-\I\bone)^{-1}}\leq\frac{C}{\abs{\xi}}.
\end{equation*}
\end{lemma}
\begin{proof}
We will show that we have a lower bound
\begin{equation}\label{lb-est}
\tgfb(\xi)\geq c \abs{\xi}^2,\quad h\xi\in[\tfrac{3\pi}{2},\tfrac{3\pi}{2}]^2.
\end{equation}
The result then follows from Lemma~\ref{lemmaG}.
We start with the estimate
\begin{align}
\tgfb(\xi)&\geq
\tfrac{4}{h^2}\sin^2(\tfrac{h}{2}\xi_1)+\tfrac{4}{h^2}\sin^2(\tfrac{h}{2}\xi_2)\notag\\
&\phantom{{}\geq{}} +\tfrac{2}{h^2}\sin(h(\xi_1-\xi_2))-\tfrac{2}{h^2}\sin(h\xi_1)
+\tfrac{2}{h^2}\sin(h\xi_2). \label{del2}
\end{align}
We have
\begin{multline*}
\tfrac{2}{h^2}\sin(h(\xi_1-\xi_2))-\tfrac{2}{h^2}\sin(h\xi_1)
+\tfrac{2}{h^2}\sin(h\xi_2)
\\
=\tfrac{2}{h^2}\bigl(
\sin(h\xi_1)(\cos(h\xi_2)-1)+\sin(h\xi_2)(1-\cos(h\xi_1))
\bigr).
\end{multline*}
We recall the elementary estimates
\begin{equation*}
\abs{\sin(\theta)}\leq\abs{\theta}\quad\text{and}\quad
1-\cos(\theta)\leq\tfrac12 \theta^2
\end{equation*}
for all $\theta\in\bR$.
Using these estimates we get
\begin{equation*}
\tfrac{2}{h^2}\abs{
\sin(h\xi_1)(\cos(h\xi_2)-1)+\sin(h\xi_2)(1-\cos(h\xi_1))}
\leq h\abs{\xi_1}\abs{\xi_2}^2+h\abs{\xi_2}\abs{\xi_1}^2.
\end{equation*}
Recall that there exists $c_0>0$ such that
$\tfrac{4}{h^2}\sin^2(\tfrac{h}{2}\xi_j)\geq c_0\abs{\xi_j}^2$ for all
$h\abs{\xi_j}\leq \frac{3\pi}{2}$, $j=1,2$. Thus using these estimates and~\eqref{del2}, we find
\begin{equation*}
\tgfb(\xi)\geq (c_0-h\abs{\xi_2})\abs{\xi_1}^2+
(c_0-h\abs{\xi_1})\abs{\xi_2}^2
\geq \tfrac12 c_0\abs{\xi}^2
\end{equation*}
for $h\abs{\xi_j}\leq\frac12 c_0$, $j=1,2$.
 
For $\frac12 c_0\leq h\abs{\xi_j}\leq \frac{3\pi}{2}$, $j=1,2$, the estimate~\eqref{lb-est} is obtained as in the proof of Lemma~\ref{lemma42}. We omit the details.
\end{proof}

\begin{lemma}\label{lemma410}
There exists $C>0$ such that
\begin{equation*}
\nbc{(G_0(\xi)-\I\bone)^{-1}-(\tGfb(\xi)-\I\bone)^{-1}}\leq
Ch
\end{equation*}
for $h\xi\in[-\tfrac{3\pi}{2},\tfrac{3\pi}{2}]^2$.
\end{lemma}
\begin{proof}
We have
\begin{equation*}
(G_0(\xi)-\I\bone)^{-1}-(\tGfb(\xi)-\I\bone)^{-1}=
(G_0(\xi)-\I\bone)^{-1}\bigl(
\tGfb(\xi)-G_0(\xi)
\bigr)(\tGfb(\xi)-\I\bone)^{-1}.
\end{equation*}
The $11$ and $22$ entries in $\tGfb(\xi)-G_0(\xi)$ are estimated by $Ch\abs{\xi}^2$; see~\eqref{11-entry}. To estimate the $12$ and $21$ entries we use Taylor's formula:
\begin{equation}\label{taylor-exp}
\e^{\I h\xi_j}=1+\I h\xi_j+(\I h\xi_j)^2
\int_0^1\e^{\I ht\xi_j}(1-t)\di t.
\end{equation}
It follows that the $12$ and $21$ entries also are estimated by $Ch\abs{\xi}^2$. Using Lemmas~\ref{lemma42} and~\ref{lemma49} the result follows.
\end{proof}

We can now state the analogue of Theorem~\ref{thm44}. The proof is omitted, since it is almost identical to the proof of Theorem~\ref{thm44}; indeed the key ingredients are the estimates in Lemmas~\ref{lemma49} and \ref{lemma410}, that correspond to the results from Lemmas~\ref{lemma42} and \ref{lemma43} with the modified symmetric difference model.
\begin{theorem}\label{thm410}
Let $K\subset (\bC\setminus \bR)\cup(-m,m)$ be compact. Then there exists $C>0$ such that
\begin{equation*}
\norm{
J_h(\tHfb-z\identh)^{-1}K_h-(H_0-z\ident)^{-1}
}_{\cB(\cH^2)} \leq C h
\end{equation*}
for all $z\in K$ and $h\in(0,1]$.
\end{theorem}

The negative result in Theorem~\ref{thm46} for the symmetric model holds also in the forward-backward case. 
\begin{lemma}\label{lemma412}
	There exists $c>0$ such that 
	\begin{equation*}
\max_{\xi\in\bT_h^2}\norm{(G_{0,h}^{\rm fb}(\xi)-\I\bone)^{-1}
-(\widetilde{G}_{0,h}^{\rm fb}(\xi)-\I\bone)^{-1}}_{\cB(\bC^2)}\geq c
	\end{equation*}
	for all $h\in(0,1]$.
\end{lemma}
\begin{proof}
As in the proof of Lemma~\ref{lemma45} we get
\begin{equation*}
\nbc{(\Gfb(\xi)-\I\bone)^{-1}
-(\tGfb(\xi)-\I\bone)^{-1}} =
\frac{f_h(\xi)}{(1+\gfb(\xi))^{1/2}(1+\tgfb(\xi))^{1/2}}.
\end{equation*}
Using~\eqref{fh},~\eqref{gfb}, and~\eqref{tgfb} we get 
	\begin{equation*}
		f_h(\tfrac{\pi}{2h},-\tfrac{\pi}{2h})=\tfrac{4}{h}, \quad \gfb(\tfrac{\pi}{2h},-\tfrac{\pi}{2h})=m^2, \quad\text{and}\quad \tgfb(\tfrac{\pi}{2h},-\tfrac{\pi}{2h})=\bigl(m+\tfrac{4}{h}\bigr)^2.
	\end{equation*}
It follows that we have a lower bound
\begin{align*}
\max_{\xi\in\bT_h^2}\norm{(\Gfb(\xi)-\I\bone)^{-1}
-(\tGfb(\xi)-\I\bone)^{-1}}_{\cB(\bC^2)}& \geq \frac{4}{(1+m^2)^{1/2}(h^2+(4+hm)^2)^{1/2}}\\
&\geq
\frac{4}{(1+m^2)^{1/2}(1+(4+m)^2)^{1/2}}
\end{align*}
for $0<h\leq1$.
\end{proof}

Theorem~\ref{thm410} combined with Lemma~\ref{lemma21} and 
Lemma~\ref{lemma412} gives the following result.
\begin{theorem}
Let $z\in(\bC\setminus\bR)\cup(-m,m)$. Then $J_h(\Hfb-z\identh)^{-1}K_h$ \emph{does not converge}
to $(H_0-z\ident)^{-1}$ in the operator norm on $\cB(\cH^2)$ as $h\to 0$.
\end{theorem}

This result implies that the strong convergence result in~\cite{SU} cannot be improved to a norm convergence result, without modifying the discretization.

\section{The 3D free Dirac operator}\label{sect3D}

Write $\sigma=(\sigma_1,\sigma_2,\sigma_3)$.
Let $\bzero$ and $\bone$ denote the $2\times2$ zero and identity matrices, and
let $\sfzero$ and $\sfone$ denote the corresponding $4\times4$ matrices.

For $U,W\in\bC^3$ there is the following identity related to the Pauli matrices, where the ``dot'' does not involve complex conjugation:
\begin{equation} \label{eq:paulicross}
	(U\cdot\sigma)(W\cdot\sigma) = (U\cdot W)\bone + \I (U\times W)\cdot \sigma.
\end{equation}

The Dirac matrices
$\alpha=(\alpha_1,\alpha_2,\alpha_3)$ and $\beta$ satisfy
\begin{align*}
	\alpha_j\alpha_k+\alpha_k\alpha_j&=2\delta_{j,k}\sfone,\\
	\alpha_j\beta+\beta\alpha_j&=\sfzero,\\
	\beta^2&=\sfone.
\end{align*}   
We can choose
\begin{equation*}\label{Dirac-choice}
	\beta=\begin{bmatrix}\bone &\bzero\\ \bzero & -\bone\end{bmatrix},\quad
	\alpha_j=\begin{bmatrix} \bzero & \sigma_j\\ \sigma_j & \bzero\end{bmatrix},\quad j=1,2,3.
\end{equation*} 

The free Dirac operator with mass $m\geq0$ in $\cH^3$ is given by
\begin{equation}\label{3DH0}
H_0=-\I\alpha\cdot\nabla+m\beta=\begin{bmatrix}
m\bone & -\I\sigma\cdot\nabla\\
-\I\sigma\cdot\nabla & -m\bone
\end{bmatrix},
\end{equation}
where $\bone$ in the context of \eqref{3DH0} denotes the identity operator on $L^2(\bR^3)\otimes \bC^2$. In Fourier space $\widehat{\cH}^3$ it is a multiplier with symbol
\begin{equation}\label{3DG0}
G_0(\xi)=\begin{bmatrix}
m\bone & \xi\cdot\sigma\\
\xi\cdot\sigma & -m\bone
\end{bmatrix}, \quad \xi\in\bR^3.
\end{equation}
Define 
\begin{equation}\label{eq53}
g_0(\xi)=m^2+\xi_1^2+\xi_2^2+\xi_3^2,
\end{equation}
then
\begin{equation}\label{eq54}
G_0(\xi)^2=g_0(\xi)\sfone.
\end{equation}
As in dimension two there are two natural discretizations of~\eqref{3DH0}, using either the pair of forward-backward partial difference operators or the symmetric partial difference operators.

\subsection{The 3D symmetric difference model}
\label{sect51}
The symmetric partial difference operators are defined in~\eqref{Djs}.
We use the notation 
\begin{equation*}
	\sfD_h^{\rm s}=(D^{\textrm{s}}_{h;1},D^{\textrm{s}}_{h;2},
	D^{\textrm{s}}_{h;3})
\end{equation*}
for the discrete symmetric gradient.
The symmetric discretization of the 3D Dirac operator is defined as
\begin{equation}
	\Hhs=\begin{bmatrix}\label{3DH0hs}
		m\bone_h & \sfD_h^{\rm s}\cdot\sigma\\
		\sfD_h^{\rm s}\cdot\sigma & -m\bone_h
	\end{bmatrix},
\end{equation}
where $\bone_h$ is the identity operator on $\ell^2(h\bZ^3)\otimes \bC^2$. In Fourier space this operator is a multiplier with symbol
\begin{equation}\label{3DG0hs}
	\Ghs(\xi)=\begin{bmatrix}
		m\bone & \sfS_h^{\rm s}(\xi)\cdot\sigma\\
		\sfS_h^{\rm s}(\xi)\cdot\sigma & -m\bone
	\end{bmatrix},
\end{equation}
where
\begin{equation*}
	\sfS_h^{\rm s}(\xi)=(\tfrac{1}{h}\sin(h\xi_1),\tfrac{1}{h}\sin(h\xi_2),\tfrac{1}{h}\sin(h\xi_3)).
\end{equation*}
We have
\begin{equation*}
	\Ghs(\xi)^2=\ghs(\xi)\sfone,
\end{equation*}
where
\begin{equation*}
	\ghs(\xi)=m^2+\tfrac{1}{h^2}\sin^2(h\xi_1)+
	\tfrac{1}{h^2}\sin^2(h\xi_2)+\tfrac{1}{h^2}\sin^2(h\xi_3).
\end{equation*}
As in the two-dimensional case we also define a modified discretization.
Let $-\Delta_h$ denote the 3D discrete Laplacian; see~\eqref{Lap}. Let $-\Delta_h\bone$ denote the $2\times 2$ diagonal operator matrix with the discrete Laplacian on the diagonal elements. Then define
\begin{equation*}
	\tHhs=\Hhs+\begin{bmatrix}
		-h\Delta_h\bone & \bzero\\
		\bzero & h\Delta_h\bone
	\end{bmatrix}.
\end{equation*}
Its symbol is
\begin{equation*}
	\tGhs(\xi)= \Ghs(\xi) + \sff_h(\xi)\begin{bmatrix}
		\bone & \bzero \\
		\bzero & -\bone
	\end{bmatrix},
\end{equation*}
where
\begin{equation} \label{eq:fh3d}
	\sff_h(\xi)=\tfrac{4}{h}\sin^2(\tfrac{h}{2}\xi_1)+
	\tfrac{4}{h}\sin^2(\tfrac{h}{2}\xi_2)+
	\tfrac{4}{h}\sin^2(\tfrac{h}{2}\xi_3).
\end{equation}
We have
\begin{equation*}
	\tGhs(\xi)^2=\tghs(\xi)\sfone,
\end{equation*}
where
\begin{equation*}
	\tghs(\xi)=\bigl(m+
	\sff_h(\xi)\bigr)^2+\tfrac{1}{h^2}\sin^2(h\xi_1)+
	\tfrac{1}{h^2}\sin^2(h\xi_2)+\tfrac{1}{h^2}\sin^2(h\xi_3).
\end{equation*}
Since $g_0$, $g_{0,h}^{\rm s}$, and $\tilde{g}_{0,h}^{\rm s}$ have similar expressions in dimensions $d = 2,3$, one can directly repeat the computations leading to Theorems~\ref{thm44} and~\ref{thm46}, which yield the following results.
\begin{theorem}\label{thm54}
	Let $K\subset (\bC\setminus\bR)\cup(-m,m)$ be compact. Then there exists $C>0$ such that
	\begin{equation*}
		\norm{
			J_h(\tHhs-z\identh)^{-1}K_h-(H_0-z\ident)^{-1}
		}_{\cB(\cH^3)} \leq C h
	\end{equation*}
	for all $z\in K$ and $h\in(0,1]$.
\end{theorem}

\begin{theorem}\label{thm55}
	Let $z\in(\bC\setminus\bR)\cup(-m,m)$. Then $J_h(\Hhs-z\identh)^{-1}K_h$ \emph{does not converge} to $(H_0-z\ident)^{-1}$ in the operator norm on $\cB(\cH^3)$ as $h\to 0$.
\end{theorem}

\subsection{The 3D forward-backward difference model} 
\label{sect52}
Using the definitions~\eqref{Dj+} and~\eqref{Dj-} we introduce the discrete forward and backward gradients as
\begin{equation*}
\sfD_h^{\pm}=(D^{\pm}_{h;1},D^{\pm}_{h;2},
D^{\pm}_{h;3}).
\end{equation*}
The forward-backward difference model is then given by
\begin{equation}\label{3DH0hfb}
\Hhfb=\begin{bmatrix}
m\bone_h & \sfD_h^-\cdot\sigma\\
\sfD_h^+\cdot\sigma & -m\bone_h
\end{bmatrix}.
\end{equation}
The symbols of $D^{\pm}_{h;j}$ in Fourier space are
\begin{equation*}
\pm\frac{1}{\I h}(\e^{\pm\I h\xi_j}-1),\quad j=1,2,3.
\end{equation*}
The symbols of the discrete gradients are then
\begin{equation} \label{eq:Spm}
\sfS^{\pm}_h(\xi)=(\pm\tfrac{1}{\I h}(\e^{\pm\I h\xi_1}-1),\pm\tfrac{1}{\I h}(\e^{\pm\I h\xi_2}-1),\pm\tfrac{1}{\I h}(\e^{\pm\I h\xi_3}-1)),
\end{equation}
such that the symbol of $\Hhfb$ is
\begin{equation*}
\Ghfb(\xi)=\begin{bmatrix}
m\bone & \sfS^-_h(\xi)\cdot\sigma\\
\sfS^+_h(\xi)\cdot\sigma & -m\bone
\end{bmatrix}.
\end{equation*}
We also define the modified discretization as
\begin{equation*}
	\tHfb=\Hfb+\begin{bmatrix}
		-h\Delta_h\bone & \bzero\\
		\bzero & h\Delta_h\bone
	\end{bmatrix}
\end{equation*}
which has the symbol
\begin{equation*}
	\tGfb(\xi)=\Ghfb(\xi) + \sff_h(\xi)\begin{bmatrix}
		\bone & \bzero \\
		\bzero & -\bone
	\end{bmatrix},
\end{equation*}
with $\sff_h$ given in \eqref{eq:fh3d}.

The arguments for norm resolvent convergence of the 3D modified forward-backward difference model do not follow as straightforwardly as in the symmetric difference case, since in particular $\tGfb(\xi)^2$ is not a diagonal matrix. A computation reveals that
\begin{equation*}
	\begin{bmatrix}
		\bzero & \sfS^-_h(\xi)\cdot\sigma \\
		\sfS^+_h(\xi)\cdot\sigma & \bzero
	\end{bmatrix}\begin{bmatrix}
		\bone & \bzero \\
		\bzero & -\bone
	\end{bmatrix} = \begin{bmatrix}
	-\bone & \bzero \\
	\bzero & \bone
\end{bmatrix}\begin{bmatrix}
\bzero & \sfS^-_h(\xi)\cdot\sigma \\
\sfS^+_h(\xi)\cdot\sigma & \bzero
\end{bmatrix}
\end{equation*}
which implies
\begin{equation} \label{eq:Gfb3dsquared}
	\tGfb(\xi)^2 = (m+\sff_h(\xi))^2\sfone + \begin{bmatrix}
		(\sfS^-_h(\xi)\cdot\sigma)(\sfS^+_h(\xi)\cdot\sigma) & \bzero \\ \bzero & (\sfS^+_h(\xi)\cdot\sigma)(\sfS^-_h(\xi)\cdot\sigma)
	\end{bmatrix}.
\end{equation}
We proceed to show the required estimates related to $\tGfb(\xi)$ in detail.

\begin{lemma}\label{lemma51}
Assume $\xi\neq0$. Then we have
\begin{equation}\label{eq57}
\nbcc{(G_0(\xi)-\I\sfone)^{-1}}\leq\frac{1}{\abs{\xi}}.
\end{equation}
There exists $C>0$ such that for  $h\xi\in[-\tfrac{3\pi}{2},\tfrac{3\pi}{2}]^3$ we have
\begin{equation}\label{eq58}
\nbcc{(\widetilde{G}_{0,h}^{\rm fb}(\xi)-\I\sfone)^{-1}}\leq \frac{C}{\abs{\xi}}.
\end{equation}
\end{lemma}
\begin{proof}
The estimate~\eqref{eq57} follows from~\eqref{eq53} and~\eqref{eq54}, together with Lemma~\ref{lemmaG}. 

To prove the estimate~\eqref{eq58} use Lemma~\ref{lemmaG} and note that 
\begin{equation*}
	\nbcc{(\widetilde{G}_{0,h}^{\rm fb}(\xi)-\I\sfone)^{-1}} = \frac{1}{\bigl(\widetilde{\lambda}_{\textup{min}}(\xi)\bigr)^{1/2}}
\end{equation*}
where $\widetilde{\lambda}_{\textup{min}}(\xi)$ is the smallest eigenvalue of $\sfone+\tGfb(\xi)^2$. Estimating $\widetilde{\lambda}_{\textup{min}}(\xi)\geq c\abs{\xi}^2$ for $h\xi\in[-\tfrac{3\pi}{2},\tfrac{3\pi}{2}]^3$ will conclude the proof.

The matrices $(\sfS^-_h(\xi)\cdot\sigma)(\sfS^+_h(\xi)\cdot\sigma)$ and $(\sfS^+_h(\xi)\cdot\sigma)(\sfS^-_h(\xi)\cdot\sigma)$ have the same spectrum, so by \eqref{eq:Gfb3dsquared} it is enough to focus on one of these blocks. Applying \eqref{eq:paulicross} to the top left block of \eqref{eq:Gfb3dsquared}, and noticing that $\sfS^+_h(\xi) = \overline{\sfS^-_h(\xi)}$, we thereby need to investigate the smallest eigenvalue of
\begin{equation*}
	\bigl(1+(m+\sff_h(\xi))^2+\abs{\sfS^-_h(\xi)}^2\bigr)\bone + \I(\sfS^-_h(\xi)\times \overline{\sfS^-_h(\xi)})\cdot\sigma.
\end{equation*}
The smallest eigenvalue of the last term is $-\abs{\sfS^-_h(\xi)\times \overline{\sfS^-_h(\xi)}}$, so we have
\begin{equation} \label{eq:lamtildemin}
	\widetilde{\lambda}_{\textup{min}}(\xi) = 1+(m+\sff_h(\xi))^2 + \abs{\sfS^-_h(\xi)}^2 - \abs{\sfS^-_h(\xi)\times \overline{\sfS^-_h(\xi)}}.
\end{equation}
If $0<\delta\leq h\abs{\xi}$ and $h\xi\in[-\tfrac{3\pi}{2},\tfrac{3\pi}{2}]^3$ for some $\delta>0$, and as $\abs{\sfS^-_h(\xi)\times \overline{\sfS^-_h(\xi)}}\leq \abs{\sfS^-_h(\xi)}^2$, then \eqref{eq:lamtildemin} and \eqref{eq:fh3d} imply the lower bound 
\begin{equation*}
	\widetilde{\lambda}_{\textup{min}}(\xi)\geq \sff_h(\xi)^2 \geq c\abs{\xi}^2.
\end{equation*}

What remains is an estimate when $h\abs{\xi}\leq\delta$ for some small $\delta>0$. Here we first need a bound on $\abs{\sfS^-_h(\xi)\times \overline{\sfS^-_h(\xi)}}$, where
\begin{equation} \label{eq:Scross}
	\sfS^-_h(\xi)\times \overline{\sfS^-_h(\xi)} = \frac{2\I}{h^2}\begin{bmatrix}
		\sin(h(\xi_3-\xi_2))-\sin(h\xi_3)+\sin(h\xi_2) \\
		\sin(h(\xi_1-\xi_3))-\sin(h\xi_1)+\sin(h\xi_3) \\
		\sin(h(\xi_2-\xi_1))-\sin(h\xi_2)+\sin(h\xi_1)
	\end{bmatrix}.
\end{equation}
Using the same estimate as in the proof of Lemma~\ref{lemma49}, we have
\begin{align*}
	\tfrac{2}{h^2}\abs{\sin(h(x-y))-\sin(hx)+\sin(hy)} \leq h\abs{x}\abs{y}^2+h\abs{y}\abs{x}^2 \leq 2\delta\abs{\xi}^2,
\end{align*}
for $h\abs{(x,y)}\leq h\abs{\xi}\leq \delta$. We have $\abs{\sfS^-_h(\xi)\times \overline{\sfS^-_h(\xi)}}\leq 2\sqrt{3}\delta\abs{\xi}^2$.

For $h\abs{\xi}\leq \delta_0 < 1$ there exists $c_0>0$ such that $\abs{\sfS^-_h(\xi)}^2\geq c_0\abs{\xi}^2$. A fixed $\delta$ with $0<\delta<\min\{\frac{c_0}{2\sqrt{3}},\delta_0\}$ gives the estimate
\begin{equation*}
	\widetilde{\lambda}_{\textup{min}}(\xi)\geq  \abs{\sfS^-_h(\xi)}^2 - \abs{\sfS^-_h(\xi)\times \overline{\sfS^-_h(\xi)}} \geq c\abs{\xi}^2
\end{equation*}
for $h\abs{\xi}\leq \delta$.
\end{proof}

\begin{lemma}\label{lemma52}
There exists $C>0$ such that
\begin{equation*}
\nbcc{(G_0(\xi)-\I\sfone)^{-1}
-(\widetilde{G}_{0,h}^{\rm fb}(\xi)-\I\sfone)^{-1}}\leq
Ch
\end{equation*}
for $h\xi\in[-\tfrac{3\pi}{2},\tfrac{3\pi}{2}]^3$.
\end{lemma}
\begin{proof}
We have
\begin{equation*}
(G_0(\xi)-\I\sfone)^{-1}
-(\widetilde{G}_{0,h}^{\rm fb}(\xi)-\I\sfone)^{-1}=
(G_0(\xi)-\I\sfone)^{-1}
\bigl(
\widetilde{G}_{0,h}^{\rm fb}(\xi)-G_0(\xi)
\bigr)
(\widetilde{G}_{0,h}^{\rm fb}(\xi)-\I\sfone)^{-1}.
\end{equation*}
We estimate the entries in $\widetilde{G}_{0,h}^{\rm fb}(\xi)-G_0(\xi)$ as in the proof of Lemma~\ref{lemma410}. Thus the entries are estimated by $Ch\abs{\xi}^2$. Using Lemma~\ref{lemma51} the result follows.
\end{proof}

Using Lemmas~\ref{lemma51} and~\ref{lemma52} we can adapt the arguments in~\cite{CGJ} to obtain the following result. We omit the details here, and refer the reader to the proof of Theorem~\ref{thm44} where details of the adaptation are given.
\begin{theorem}\label{thm53}
Let $K\subset (\bC\setminus\bR)\cup(-m,m)$ be compact. Then there exists $C>0$ such that
\begin{equation*}
\norm{
J_h(\widetilde{H}_{0,h}^{\rm fb}-z\identh)^{-1}K_h-(H_0-z\ident)^{-1}
}_{\cB(\cH^3)} \leq C h
\end{equation*}
for all $z\in K$ and $h\in(0,1]$.
\end{theorem}

As in dimension two, the unmodified forward-backward difference model does not lead to norm resolvent convergence.
\begin{lemma}\label{lemma56}
There exists $c>0$ such that 
\begin{equation*}
\max_{\xi\in\bT_h^3}\norm{
(G_{0,h}^{\rm fb}(\xi)-\I\sfone)^{-1}
-(\widetilde{G}_{0,h}^{\rm fb}(\xi)-\I\sfone)^{-1}}_{\cB(\bC^4)}\geq c
\end{equation*}
for all $h\in(0,1]$.
\end{lemma}
\begin{proof}
	Consider $\xi_h = (\frac{\pi}{2h},-\frac{\pi}{2h},0)$ and notice that $\xi_h\in\bT_h^3$ for $h>0$. From \eqref{eq:fh3d}, \eqref{eq:Spm}, and \eqref{eq:Scross} then
	\begin{equation*}
		\sff_h(\xi_h) = \tfrac{4}{h}\quad \text{and}\quad \abs{\sfS^-_h(\xi_h)}^2 = \abs{\sfS^-_h(\xi_h)\times\overline{\sfS^-_h(\xi_h)}} = \tfrac{4}{h^2}.
	\end{equation*}
	From \eqref{eq:lamtildemin}, the smallest eigenvalues $\widetilde{\lambda}_{\textup{min}}(\xi_h)$ of $\sfone+\widetilde{G}_{0,h}^{\rm fb}(\xi_h)^2$ and $\lambda_{\textup{min}}(\xi_h)$ of $\sfone+G_{0,h}^{\rm fb}(\xi_h)^2$ are
	\begin{align*}
		\widetilde{\lambda}_{\textup{min}}(\xi_h) &= 1+(m+\sff_h(\xi_h))^2 + \abs{\sfS^-_h(\xi_h)}^2 - \abs{\sfS^-_h(\xi_h)\times \overline{\sfS^-_h(\xi_h)}} = 1+\bigl(m+\tfrac{4}{h}\bigr)^2, \\
		\lambda_{\textup{min}}(\xi_h) &= 1+m^2 + \abs{\sfS^-_h(\xi_h)}^2 - \abs{\sfS^-_h(\xi_h)\times \overline{\sfS^-_h(\xi_h)}} = 1+m^2.
	\end{align*}
	Thus by Lemma~\ref{lemmaG} we have 
	\begin{equation*}
		\nbcc{(\widetilde{G}_{0,h}^{\rm fb}(\xi_h)-\I\sfone)^{-1}} = \frac{1}{\bigl(\widetilde{\lambda}_{\textup{min}}(\xi_h)\bigr)^{1/2}} = \frac{h}{(h^2+(mh+4)^2)^{1/2}}
	\end{equation*}
	and
	\begin{equation*}
		\nbcc{(G_{0,h}^{\rm fb}(\xi_h)-\I\sfone)^{-1}} = \frac{1}{\bigl(\lambda_{\textup{min}}(\xi_h)\bigr)^{1/2}} = \frac{1}{(1+m^2)^{1/2}},
	\end{equation*}
	which conclude the proof with 
	\begin{equation*}
		c = \frac{1}{(1+m^2)^{1/2}}-\frac{1}{(1+(m+4)^2)^{1/2}}>0. \qedhere
	\end{equation*}
\end{proof}

Combining Theorem~\ref{thm53} with Lemma~\ref{lemma21} and Lemma~\ref{lemma56} gives the following non-convergence result.
\begin{theorem}
	Let $z\in(\bC\setminus\bR)\cup(-m,m)$. Then $J_h(\Hfb-z\identh)^{-1}K_h$ \emph{does not converge}
	to $(H_0-z\ident)^{-1}$ in the operator norm on $\cB(\cH^3)$ as $h\to 0$.
\end{theorem}

\section{Sobolev space estimates and strong convergence} \label{sectSobolev}
In sections~\ref{sect1D}--\ref{sect3D} we have shown that
$J_h(H_{0,h}-zI_h)^{-1}K_h$ converges in the $\cB(\cH^d)$-operator norm to $(H_0-zI)^{-1}$ for several choices of discrete model $H_{0,h}$, and we have also shown that in other cases this norm convergence does not hold. This section is dedicated to the cases where $J_h(H_{0,h}-zI_h)^{-1}K_h$ does not converge to $(H_0-zI)^{-1}$ in the $\cB(\cH^d)$-operator norm, and instead we will prove that convergence holds in the $\cB(H^1(\bR^d)\otimes\bC^{\nu(d)},\cH^d)$-operator norm. These latter results obviously imply strong convergence. In particular, we recover the result in~\cite{SU} for $d=2$ with the discretization $H_{0,h}^{\rm fb}$.


\subsection{The 1D model}
The 1D symmetric model 
$H_{0,h}^{\rm s}$ is defined in~\eqref{1DH0hs} and its symbol $G_{0,h}^{\rm s}(\xi)$ in~\eqref{1DG0hs}. The symbol for the continuous Dirac operator $G_0(\xi)$ is defined in~\eqref{1DG0}.

\begin{lemma}\label{lemma61}
There exists $C>0$ such that
\begin{equation*}
\bigl\lVert
\bigl((G_{0,h}^{\rm s}(\xi)-\I\bone)^{-1}-(G_0(\xi)-\I\bone)^{-1}\bigr)
(G_0(\xi)-\I\bone)^{-1}\bigr\rVert_{\cB(\bC^2)}
\leq Ch
\end{equation*}
for $h\xi\in[-\tfrac{3\pi}{2},\tfrac{3\pi}{2}]$. 
\end{lemma}
\begin{proof}
Note that Lemma~\ref{lemmaG} and~\eqref{1Dg0hs} imply the estimate $\nbc{(G_{0,h}^{\rm s}(\xi)-\I\bone)^{-1}}\leq (1+m^2)^{-1}$ and that this estimate cannot be improved for $h\xi\in[-\tfrac{3\pi}{2},\tfrac{3\pi}{2}]$.
We have
\begin{align*}
\bigl((G_0(\xi)-\I\bone)^{-1}-(&G_{0,h}^{\rm s}(\xi)-\I\bone)^{-1}\bigr)(G_0(\xi)-\I\bone)^{-1}
\\
&=
(G_{0,h}^{\rm s}(\xi)-\I\bone)^{-1}\bigl(
G_{0,h}^{\rm s}(\xi)-G_0(\xi)
\bigr)(G_0(\xi)-\I\bone)^{-2}.
\end{align*}
Now
\begin{equation*}
G_{0,h}^{\rm s}(\xi)-G_0(\xi)=\begin{bmatrix}
0 & \tfrac{1}{h}\sin(h\xi)-\xi\\
\tfrac{1}{h}\sin(h\xi)-\xi & 0
\end{bmatrix}.
\end{equation*}
Using Taylor's formula~\eqref{taylor-sin} together with the estimates
\begin{equation*}
\nbc{(G_0(\xi)-\I\bone)^{-1}}\leq\frac{1}{\abs{\xi}}\quad\text{and}\quad
\nbc{(G_{0,h}^{\rm s}(\xi)-\I\bone)^{-1}}\leq 1
\end{equation*}
we get
\begin{align*}
\bigl\lVert
\bigl((G_{0,h}^{\rm s}(\xi)-\I\bone)^{-1}-(G_0(\xi)-\I\bone)^{-1}\bigr)&
(G_0(\xi)-\I\bone)^{-1}\bigr\rVert_{\cB(\bC^2)}
\leq Ch^2\abs{\xi}\leq Ch
\end{align*}
for $h\xi\in[-\tfrac{3\pi}{2},\tfrac{3\pi}{2}]$.
\end{proof}

\begin{proposition}\label{prop62}
Let $K\subset(\bC\setminus\bR)\cup(-m,m)$ be compact. Then there exists $C>0$ such that
\begin{equation}\label{H0hs-est}
\norm{
J_h(H_{0,h}^{\rm s}-z\identh)^{-1}K_h-(H_0-z\ident)^{-1}
}_{\cB(\Sobone,\cH^1)}\leq C h
\end{equation}
for all $z\in K$ and $h\in(0,1]$.
\end{proposition}
\begin{remark}
The estimate~\eqref{H0hs-est} implies
\begin{equation*}
\slim_{h\to0}J_h(H_{0,h}^{\rm s}-z\identh)^{-1}K_h=(H_0-z\ident)^{-1}
\end{equation*}
uniformly in $z\in K$.
\end{remark}
\begin{proof}
The result follows if we prove the estimate
\begin{equation*}
\norm{\bigl(
J_h(H_{0,h}^{\rm s}-z\identh)^{-1}K_h-(H_0-z\ident)^{-1}
\bigr)(H_0-\I I)^{-1}}_{\cB(\cH^1)}\leq C h.
\end{equation*}
The proof is very similar to the proof of Theorem~\ref{thm44}. In the arguments one replaces $\cF^{\ast}u$ by $\cF^{\ast}(H_0-\I I)^{-1}u$ and uses Lemma~\ref{lemma61}. Further details are omitted.
\end{proof}

\subsection{The 2D model}
Consider first the symmetric difference model. The continuous 2D free Dirac operator is denoted by $H_0$ and its symbol by $G_0(\xi)$; see~\eqref{2DG0}. The symmetric difference model is denoted by $H_{0,h}^{\rm s}$; see~\eqref{2DH0h}. Its symbol is denoted by $G_{0,h}^{\rm s}(\xi)$; see~\eqref{2DG0h}.
\begin{lemma}\label{lemma64}
There exists $C>0$ such that
\begin{equation*}
\bigl\lVert\bigl((G_{0,h}^{\rm s}(\xi)-\I\bone)^{-1}
-(G_0(\xi)-\I\bone)^{-1}\bigr)
(G_0(\xi)-\I\bone)^{-1}\bigr\rVert_{\cB(\bC^2)}
\leq Ch
\end{equation*}
for $h\xi\in[-\tfrac{3\pi}{2},\tfrac{3\pi}{2}]^2$. 
\end{lemma}
\begin{proof}
The proof is almost the same as the proof of Lemma~\ref{lemma61}.
It follows from Lemma~\ref{lemmaG} and~\eqref{g} that we have
$\nbc{(G_{0,h}^{\rm s}(\xi)-\I\bone)^{-1}}\leq 1$ for $h\xi\in[-\tfrac{3\pi}{2},\tfrac{3\pi}{2}]^2$. We have
\begin{multline*}
G_{0,h}^{\rm s}(\xi)-G_0(\xi)=\\
\begin{bmatrix}
0 & (\tfrac{1}{h}\sin(h\xi_1)-\xi_1)-{\I} (\tfrac{1}{h}\sin(h\xi_2)-\xi_2)\\
(\tfrac{1}{h}\sin(h\xi_1)-\xi_1)+{\I} (\tfrac{1}{h}\sin(h\xi_2)-\xi_2) &0
\end{bmatrix}.
\end{multline*}
Then~\eqref{taylor-sin} implies
$\nbc{G_{0,h}^{\rm s}(\xi)-G_0(\xi)}\leq C h^2\abs{\xi}^3$. The remaining parts of the argument in the proof of Lemma~\ref{lemma61} can then be repeated.
\end{proof}

The proof of the next result is almost identical to the proof of Proposition~\ref{prop62} and is omitted.
\begin{proposition}\label{prop65}
Let $K\subset(\bC\setminus\bR)\cup(-m,m)$ be compact. Then there exists $C>0$ such that
\begin{equation*}\label{H0hs-est-2D}
\norm{
J_h(H_{0,h}^{\rm s}-z\identh)^{-1}K_h-(H_0-z\ident)^{-1}
}_{\cB(\Sobtwo,\cH^2)}\leq C h
\end{equation*}
for all $z\in K$ and $h\in(0,1]$.
\end{proposition}

Next we consider the forward-backward difference model. The arguments are almost identical to those for the symmetric model. We state the result without proof. The discretization $H_{0,h}^{\rm fb}$ is defined in~\eqref{2DH0fb}.
\begin{proposition}\label{prop66}
Let $K\subset(\bC\setminus\bR)\cup(-m,m)$ be compact. Then there exists $C>0$ such that
\begin{equation*}\label{H0hfb-est-2D}
\norm{
J_h(H_{0,h}^{\rm fb}-z\identh)^{-1}K_h-(H_0-z\ident)^{-1}
}_{\cB(\Sobtwo,\cH^2)}\leq C h
\end{equation*}
for all $z\in K$ and $h\in(0,1]$.
\end{proposition}

\subsection{The 3D model} 
In this section $H_0$ denotes the free 3D Dirac operator. The symmetric difference model $H_{0,h}^{\rm s}$ is given by~\eqref{3DH0hs} and the forward-backward difference model $H_{0,h}^{\rm fb}$ is given by~\eqref{3DH0hfb}.

We state the following results without proofs, since they are very similar to the proofs of Propositions~\ref{prop62}, \ref{prop65}, and \ref{prop66}.
\begin{proposition}
Let $K\subset(\bC\setminus\bR)\cup(-m,m)$ be compact. Then there exists $C>0$ such that
\begin{equation*}
\norm{
J_h(H_{0,h}^{\rm s}-z\identh)^{-1}K_h-(H_0-z\ident)^{-1}
}_{\cB(\Sobthree,\cH^3)}\leq C h
\end{equation*}
and
\begin{equation*}
	\norm{
		J_h(H_{0,h}^{\rm fb}-z\identh)^{-1}K_h-(H_0-z\ident)^{-1}
	}_{\cB(\Sobthree,\cH^3)}\leq C h
\end{equation*}
for all $z\in K$ and $h\in(0,1]$.
\end{proposition}

\section{Perturbed Dirac operators}\label{sectV}
In this section we state results on perturbed Dirac operators and their discretizations, with respect to norm resolvent convergence.
We use the following condition on the perturbation. 

\begin{assumption}\label{assump71}
Assume that $V\colon\bR^d\to\cB(\bC^{\nu(d)})$ is bounded and Hölder continuous with exponent $\theta\in(0,1]$. Assume $(V(x))^{\ast}=V(x)$, $x\in\bR^d$.
\end{assumption}
We require another assumption on $\psi_0$ in addition to Assumption~\ref{assumpA2}. We emphasize that concrete examples of $\psi_0$ satisfying these assumptions are given in~\cite[subsection~2.1]{CGJ}.
\begin{assumption}\label{assump72}
Assume there exists $\tau>d$ such that
\begin{equation*}
\abs{\psi_0(x)}\leq (1+\abs{x})^{-\tau},\quad x \in\bR^d.
\end{equation*}
\end{assumption}
Define a discretization of $V$ by 
\begin{equation}\label{Vh}
V_h(k)=V(hk),\quad k\in\bZ^d.
\end{equation}
Let $H_{0,h}$ be one of the discretizations of $H_0$ from sections~\ref{sect1D}--\ref{sect3D}. Then we define self-adjoint operators $H=H_0+V$ and $H_h=H_{0,h}+V_h$.

The following result is an adaptation of~\cite[Proposition~4.3]{CGJ}
to the present framework.
\begin{lemma} \label{lemma73}
Let $V$ satisfy Assumption~\ref{assump71} and let $\psi_0$ satisfy Assumption~\ref{assump72}. Define
\begin{equation}\label{theta'}
\frac{1}{\theta'}=\frac{1}{\theta}+\frac{1}{\tau-d}.
\end{equation}
Then
\begin{equation*}
\norm{V_hK_h-K_hV}_{\cB(\cH^d,\cH^d_h)}\leq C h^{\theta'}.
\end{equation*}
\end{lemma}
\begin{proof}
The proof in~\cite{CGJ} can be directly adapted to the current framework. We omit the details. Note that $\psi_0(x)$ is a scalar, such that we have $\psi_0(x)V(x)f(x)=V(x)\psi_0(x)f(x)$, $f\in\cH^d$.
\end{proof}

We can then state our main result on the perturbed Dirac operators, which follows from Lemma~\ref{lemma73} and a direct adaptation of the proof of \cite[Theorem~4.4]{CGJ}.

\begin{theorem}
Let $J_h$ and $K_h$ be the operators defined in section~\ref{sectPrelim}, and let $\psi_0$ satisfy Assumption~\ref{assump72}.
Let $V$ satisfy Assumption~\ref{assump71} and define $V_h$ by~\eqref{Vh}. Let $H_{0,h}$ equal either of
\begin{equation*}
\begin{cases}
H^{\rm fb}_{0,h}, & d=1,\\
\widetilde{H}^{\rm fb}_{0,h},& d=2,3,\\
\widetilde{H}^{\rm s}_{0,h},& d=1,2,3.
\end{cases}
\end{equation*}
Let $H_h=H_{0,h}+V_h$. Let $H=H_0+V$, where $H_0$ is the free Dirac operator in the relevant dimension.
Assume $V\not\equiv0$ and let $\theta'$ be given by~\eqref{theta'}. Then the following result holds.

Let $K\subset \bC\setminus\bR$ be compact. Then there exist $C>0$ and $h_0>0$ such that
\begin{equation*}
\norm{
J_h(H_{h}-z\identh)^{-1}K_h-(H-z\ident)^{-1}
}_{\cB(\cH^d)} \leq C h^{\theta'}
\end{equation*}
for all $z\in K$ and $h\in(0,h_0]$.
\end{theorem} 

\paragraph*{Acknowledgments.} This research is partially supported by grant 8021--00084B from Independent Research Fund Denmark \textbar\ Natural Sciences.

\end{document}